\documentclass[12pt]{article}

\usepackage[textwidth=125mm, textheight=195mm]{geometry}

\usepackage{amsmath}
\usepackage{amssymb}
\usepackage{amsthm}

\usepackage{slashed}
\usepackage{color}
\usepackage{esint}

\usepackage{graphicx}
\usepackage[all,cmtip]{xy}
\usepackage{hyperref}

\def\cutint{{\int \!\!\!\!\!\! -}}
\newtheorem{thm}{Theorem}
\newtheorem{lemma}{Lemma}
\newtheorem{prop}{Proposition}

\newtheorem{defn}{Definition}

\newtheorem{rmk}{Remark}

\hypersetup{
    pdfborder = {0 0 0},%
    colorlinks,%
    citecolor=blue,%
    filecolor=black,%
    linkcolor=red,%
    urlcolor=green
}

\begin{document}

\title{ Noncommutative Chern--Simons theory on the quantum 3-sphere $S^3_\theta$}
\author{ 
        \textsc{Dan Li} \\
        \\
	    Department of Mathematics, Purdue University \\
		150 N. University St, West Lafayette, IN 47907 \\
		E-mail address: li1863@math.purdue.edu
        58B34, 58J42, 81R60, 81T75   
                  spectral triple, noncommutative geometry \\
	    }
\date{}
\maketitle

\begin{abstract}
We consider the $\theta$-deformed quantum 3-sphere $S^3_\theta$ and study its
Chern--Simons theory from a spectral point of view.
We first construct a spectral triple on $S^3_\theta$ as a generalization of the Dirac geometry on $S^3 $. 
Since the choice of Dirac operator is not unique, we give two more natural spectral triples 
on $S^3_\theta$ related to the round metric.
We then compute the Chern--Simons action with respect to these three spectral triples, 
it turns out that it is not  a topological invariant,
that is, it depends on the choice of Dirac operators.

\end{abstract}

\newpage

\section{Introduction}\label{Intro}

In order to understand the abstract theory of noncommutative geometry, 
it is better to test it on concrete examples.
Besides the well-known noncommutative torus, it is natural to consider noncommutative spheres.
Indeed, there are a variety of quantum  spheres proposed by authors from different points of view in the literature \cite{D03}.

Our main object to study in this paper is the quantum 3-sphere $S^3_\theta$, which was first introduced by Connes and Landi 
in \cite{CL01} from a K-theoretic
consideration. In fact,
$S^3_\theta$ is a special case of a more general class of 
noncommutative 3-spheres considered in \cite{CD02}, and  it also coincides with the quantum spheres discussed in \cite{M91, NO97}.
In physics, $S^3_\theta$ has possible applications in condensed matter physics and quantum gravity. 

By definition $S^3_\theta$ is a $\theta$-deformed $C^*$-algebra
and its K-theory is known to be $K_0(S^3_\theta) \cong \mathbb{Z}, \, K_1(S^3_\theta) \cong \mathbb{Z}$.
 The 4-dimensional quantum sphere $S^4_\theta$ considered in \cite{CL01} is the suspension of $S^3_\theta$, so it is possible to obtain
 the Dirac operator on $S^3_\theta$ by dimensional reduction from that on $S^4_\theta$. The quantum  3-sphere $S^3_\theta$ 
is similar to $SU_q(2)$ $\,(0 < q < 1)$, but now with 
a real parameter $\theta \in \mathbb{R}$ 
(or equivalently a complex parameter $\lambda = e^{2 \pi i \theta} \in U(1)$), 
and the noncommutative 2-torus $\mathbb{T}^2_\theta$ is naturally embedded inside $S^3_\theta$.

The Chern--Simons form was first introduced in \cite{CS74} as a boundary term when the authors were computing the first Pontryagin number of 
a 4-manifold. It can be defined as a secondary characteristic class  by the 
transgression of the Chern character on principal bundles. 
Let $M$ be a closed oriented 3-manifold and $G$ a  simply connected compact
 Lie group, for example $SU(2)$, with Lie algebra $\mathfrak{g}$.
 If $P \rightarrow M$ is a principal $G$-bundle and $A \in \mathcal{A}_P \subset \Omega^1_P(\mathfrak{g})$ 
 is a $\mathfrak{g}$-valued connection 1-form on $P$, 
then the Chern--Simons action is defined by the integral,
\begin{equation}
   CS(A) =  \frac{1}{8\pi^2}  \int_M tr \left( A\wedge dA + \frac{2}{3}A\wedge A \wedge A \right)
\end{equation}
where $tr$ is an invariant bilinear form on $\mathfrak{g}$. 
Under a gauge transformation
\begin{equation*}
   A \mapsto A^g = g^{-1}A g + g^{-1}dg, \quad g: M \rightarrow G
\end{equation*}
the Chern--Simons action is gauge invariant up to an integral winding number, 
\begin{equation}
   CS(A^g) = CS(A) + \frac{1}{24\pi^2}\int_M tr (g^{-1}dg)^3
\end{equation}
The classical Chern-Simons form is a topological invariant in the sense that it is independent of the background metric. The study of 
quantum Chern-Simons theory \cite{W89} lies at the intersection of many fields such as 
quantum topology, quantum topological field theory and conformal field theory etc.

There are different proposals for the definition of Chern--Simons action in noncommutative geometry
and the difficulty was in its gauge invariance. 
For instance,   Chamseddine and Fr{\"o}hlich \cite{CF94} defined the noncommutative Chern--Simons action based on the idea of transgression, 
Krajewski \cite{K98} used the Dixmier trace instead of the classical integral over 3-manifolds. 
Connes and Chamseddine  introduced the Chern--Simons action  as the integral relative to a cyclic 3-cocycle in  \cite{CC06},
they obtained the variation of the spectral action under inner fluctuations
 as a Yang-Mills action plus a Chern--Simons action assuming that the tadpole graph does not contribute. The above mentioned 
noncommutative Chern--Simons actions are not gauge invariant in general. 
 
In \cite{P1201}, Pfante  gave a definition
of noncommutative Chern--Simons action for 3-summable spectral triples, 
which is gauge invariant up to a Fredholm index based on the local index formula \cite{CM95}. In this new action,
besides a 3-cocycle $\phi_3$ there also exists a 1-cocycle $\phi_1$
so that $(\phi_1, \phi_3)$ forms a $(b, B)$-cocycle, when $\phi_1$ vanishes it coincides with Connes and Chamseddine's definition.
 Pfante computed the Chern--Simons action over the quantum compact group $SU_q(2)$ \cite{P1201} 
and the noncommutative 3-torus $\mathbb{T}^3_\Theta$  \cite{P1202} as examples. In the case of $SU_q(2)$ the $\phi_1$ term contributes to the action, 
while on $\mathbb{T}^3_\Theta$ the $\phi_1$ term vanishes. 

In this paper we first recall the noncommutative local index formula and the definition of Chern--Simons action in section \ref{NCaction}.
After introducing the quantum 3-sphere $S^3_\theta$, we explicitly construct the first spectral triple generalizing 
the Dirac geometry on $S^3$ in section \ref{QuanSp}. The dimension spectrum of this spectral triple is discussed in section \ref{CSaction}
and further its Chern--Simons action is computed, in particular the linear term $\phi_1$ vanishes on $S^3_\theta$. 
However, there are two more natural Dirac operators on  $S^3_\theta$
related to the round metric on $S^3$, one is a reduction from the Dirac operator on $S^4_\theta$ and the other is defined based on 
another orthogonal framing in Hopf fibration. In section \ref{Diracop} we compute the Chern--Simons action with respect to these 
 spectral triples  
and different Chern--Simons actions are compared. It turns out that the choice of  Dirac operator determines the Chern--Simons action, 
we conclude that the noncommutative Chern--Simons action is not a topological invariant on  $S^3_\theta$, which was also observed on $SU_q(2)$ in \cite{P1201}.

\section{Local index formula}\label{NCaction}
In order to fix the notations  we briefly  recall   
the noncommutative local index formula in  three dimensions \cite{CM95, H06} 
and the definition of noncommutative Chern--Simons action following \cite{P1201}. The local index formula on $SU_q(2)$ has been studied
in \cite{C04, SDLSV05} and the Chern--Simons action on $SU_q(2)$ is discussed in \cite{P1201} as well.

A noncommutative odd-dimensional Riemannian manifold is described by an odd spectral triple $(\mathcal{A}, \mathcal{H}, \mathcal{D})$.
$\mathcal{A}$ is a unital associative algebra with involution, in practice $\mathcal{A}$ is usually a pre-$C^*$-algebra 
closed under holomorphic functional calculus. 
$\mathcal{A}$ acts on the separable Hilbert space $\mathcal{H}$ as bounded operators through  a faithful representation $\pi :  \mathcal{A} \rightarrow B(\mathcal{H})$.  
$\mathcal{D}$ is an unbounded self-adjoint operator with compact resolvent
such that $[\mathcal{D}, a]$ is bounded for any $a \in \mathcal{A}$. 
Furthermore, the Dirac-type operator $\mathcal{D}$ determines the metric on the state space of $\mathcal{A}$,
\begin{equation*}
  d(\phi, \psi) = sup_{a \in \mathcal{A}} \{ |\phi(a) -  \psi(a)|;\, || [\mathcal{D}, a]|| \leq 1 \}
\end{equation*}
The prototype of a spectral triple is given by $(C^\infty(M), L^2_g(M, \slashed{S}), \slashed{D}_g)$, i.e. the Dirac geometry 
on a closed Riemannian spin manifold $(M, g)$. 

Given an odd spectral triple $(\mathcal{A}, \mathcal{H}, \mathcal{D})$, let $\mathcal{F} = \mathcal{D} |\mathcal{D}|^{-1}$ 
be the sign of $\mathcal{D}$ such that $\mathcal{F}^2 = 1$, and $P = (\mathcal{F} +1)/2$ be the projection onto the $+1$ eigenspace of $\mathcal{F}$ in $\mathcal{H}$.  
For a unitary operator $u \in U(\mathcal{A})$, $PuP: P\mathcal{H} \rightarrow P\mathcal{H}$ is a Fredholm operator with its analytic index defined by,
\begin{equation}\label{FredInd}
   Index(PuP) = dim ker PuP - dim ker Pu^*P
\end{equation}
In other words, the spectral triple $(\mathcal{A}, \mathcal{H}, \mathcal{D})$ determines an additive map by the Fredholm index,
\begin{equation} 
        ind_\mathcal{D} :    K_1(\mathcal{A}) \rightarrow \mathbb{Z}; \quad [u] \mapsto Index(PuP)
\end{equation}

 The triple $(\mathcal{A}, \mathcal{H}, \mathcal{F})$ is called the associated Fredholm module over $\mathcal{A}$, which can be viewed as
an abstract elliptic operator in K-homology.  $(\mathcal{A}, \mathcal{H}, \mathcal{F})$ is called $p$-summable if for every integer
$n \geq p$ the following product is in the trace class $\mathcal{L}^1 \subset \mathcal{K}$,
\begin{equation*}
   [\mathcal{F}, a_0] [\mathcal{F}, a_1] \cdots [\mathcal{F}, a_n] \in \mathcal{L}^1, \quad \forall \,a_i \in \mathcal{A}
\end{equation*}
A relevant concept is the dimension of the spectral triple $(\mathcal{A}, \mathcal{H}, \mathcal{D})$, which is defined as the 
smallest integer $p$ such that the characteristic values $\mu_n$ of $\mathcal{D}^{-1}$ 
behave like  
\begin{equation*}
   \mu_n(\mathcal{D}^{-1}) = O(n^{-1/p}), \quad \text{as} \quad n \rightarrow \infty
\end{equation*}
If the dimension $p$ of $(\mathcal{A}, \mathcal{H}, \mathcal{D})$ is finite, then it is called a $p$-summable spectral triple.
In particular, the dimension of the Dirac geometry $(C^\infty(M), L^2(M, \slashed{S}), \slashed{D})$ equals the dimension of the base manifold $M$.

On the other hand, the Fredholm index \eqref{FredInd} can also be computed by pairing $K_1(\mathcal{A}) $ with the odd Connes--Chern character of a Fredholm module in cyclic cohomology.
Denote by $C^n(\mathcal{A})$ the space of
$(n+1)$-linear functionals $\phi: \mathcal{A}^{\otimes {n+1}} \rightarrow \mathbb{C}$ such that $\phi(a_0, a_1, \cdots, a_n) = 0$ if 
$a_j = 1$ for some $j \geq 1$. The coboundary map $b: C^{n}(\mathcal{A}) \rightarrow C^{n+1}(\mathcal{A})$ is defined by 
\begin{equation*}
   \begin{array}{ll}
         b\phi(a_0, \cdots, a_{n+1}) = & \sum_{j=0}^n (-1)^j \phi(a_0, \cdots, a_j a_{j+1}, \cdots, a_{n+1}) \\
                                      &  + (-1)^{n+1} \phi(a_{n+1}a_0, \cdots, a_n)
   \end{array}
\end{equation*}
Since $b^2=0$, one defines the Hochschild cohomology groups of $\mathcal{A}$ by the cohomology of the Hochschild complex $(C^*(\mathcal{A}), b)$, 
denoted by $HH^n(\mathcal{A})$.
In addition,  $\phi: \mathcal{A}^{\otimes {n+1}} \rightarrow \mathbb{C}$ is said to be cyclic if $\phi = \lambda \phi$, where
\begin{equation*}
    \lambda \phi(a_0, \cdots, a_{n}) =  (-1)^n \phi(a_n, a_0, \cdots, a_{n-1}) 
\end{equation*}
Then one has the cyclic complex, denoted by $(C^*_\lambda(\mathcal{A}), b)$, as a subcomplex of  $(C^*(\mathcal{A}), b)$, similarly one
defines the cyclic cohomology groups $HC^n(\mathcal{A})$.

\begin{thm}[\cite{C94}]
  Let $(\mathcal{A}, \mathcal{H}, \mathcal{F})$ be a $p$-summable odd Fredholm module and 
  let $n = 2k+1 \geq p$, then   the following cochain
   \begin{equation}
     \begin{array}{ll}
        \phi(a_0, \cdots, a_n) & = \frac{1}{2} Tr(\mathcal{F}[\mathcal{F}, a_0][\mathcal{F}, a_1] \cdots [\mathcal{F}, a_n]) \\
                               & = Tr( a_0[\mathcal{F}, a_1] \cdots [\mathcal{F}, a_n])
     \end{array}
  \end{equation}
 defines a cyclic cocycle such that its pairing with a unitary $u \in U({\mathcal{A}})$ computes the Fredholm index up to a normalization constant,
  \begin{equation}
     \phi(u, u^*, \cdots, u, u^*) = (-1)^{(n+1)/2} 2^nIndex(PuP)
  \end{equation}
\end{thm}

One step further, one uses the periodic cyclic cohomology groups to pair with K-groups. 
Define Connes' boundary map by the composition $B = N \circ B_0: C^n(\mathcal{A}) \rightarrow C^{n-1}(\mathcal{A}) \rightarrow C^{n-1}(\mathcal{A})$, more precisely,
\begin{equation*}
   \begin{array}{ll}
         B_0\phi(a_0, \cdots, a_{n-1}) =  \phi(1, a_0,  \cdots, a_{n-1}) - (-1)^{n} \phi(a_0, \cdots, a_{n-1}, 1); \\
           N\phi(a_0, \cdots, a_{n-1}) = \sum_{j =0}^{n-1} \lambda^{j} \phi = \sum (-1)^{(n-1)j} \phi(a_j, a_{j+1}, \cdots, a_{j-1});\\
           B\phi(a_0, \cdots, a_{n-1}) =  \sum_{j=0}^n (-1)^{nj} \phi(1, a_j, a_{j+1}, \cdots, a_{j-1})
   \end{array}
\end{equation*}
Since $b^2 = Bb + bB = B^2 = 0$, one has the Connes' $(b, B)$-bicomplex, denoted by $B(\mathcal{A})$, and define the periodic
cyclic cohomology $HP^*(\mathcal{A})$ as the cohomology of the total complex $(TotB(\mathcal{A}), b+B )$. 

For example, an odd $(b, B)$-cocycle $\phi \in HP^1(\mathcal{A})$ is defined by 
\begin{equation*}
   \phi = (\phi_1, \phi_3, \phi_5, \cdots), \quad s.t. \quad b\phi_{2k-1}+ B\phi_{2k+1} = 0
\end{equation*}
In fact, the map 
\begin{equation*}
   \begin{array}{ccl}
      (C^*_\lambda(\mathcal{A}), b)  & \rightarrow & (TotB(\mathcal{A}), b+B ) \\
       \phi_n & \mapsto & (0, \cdots, 0, \phi_n, 0 , \cdots )
   \end{array}
\end{equation*}
induces  a quasi-isomorphism of complexes.  

\begin{defn} Let $(\mathcal{A}, \mathcal{H}, \mathcal{F})$ be a $p$-summable odd Fredholm module and let $n = 2k+1 \geq p$,  
  its odd Connes--Chern character  is defined by
  \begin{equation}\label{CCCh}
     Ch_{n}(a_0, \cdots, a_{n}) = \frac{\Gamma(1+ n/2)}{2 \cdot n!} Tr(\mathcal{F} [\mathcal{F}, a_0][\mathcal{F}, a_1] \cdots [\mathcal{F}, a_n]),
  \end{equation}
  which is a cyclic cocycle and its periodic cyclic cohomology class is independent of  the choice of $n$.
\end{defn}

A spectral triple $(\mathcal{A}, \mathcal{H}, \mathcal{D})$ is regular if 
 $\mathcal{A}$ and
$[\mathcal{D}, \mathcal{A}]$ both belong to  $ OP^0 = \cap_{n \geq 1} Dom \delta^n$, i.e., 
the domain of all derivations $\delta^n$ with  $\delta(T) := [|\mathcal{D}|, T]$. 
Let $\mathcal{L}^{1, \infty}$ be the set of compact operators having finite  $\| \cdot \|_{1,\infty} $-norm, where
\begin{equation*}
   \|T\|_{1,\infty} = \sup_N\frac{\sum_{i=1}^N \mu_i(T)}{\log N}
\end{equation*}
There exists  a well-defined trace functional on $\mathcal{L}^{1, \infty}$, i.e., the Dixmier trace $Tr_\omega: \mathcal{L}^{1, \infty} \rightarrow \mathbb{C}$. 

The Hochschild character theorem tells us that we can compute the Connes--Chern character \eqref{CCCh} by a Hochschild cohomology class.
\begin{thm}[\cite{C94}] Let $(\mathcal{A}, \mathcal{H}, \mathcal{D})$ be a regular odd spectral triple, 
assume $a \cdot |D|^{-n} \in \mathcal{L}^{1, \infty}$ for every $a \in \mathcal{A}$ and some odd positive integer $n$, 
 then the Connes--Chern character \eqref{CCCh} is cohomologous to the Hochschild cocycle
 \begin{equation}
    \Phi (a_0, \cdots, a_n) = \frac{\Gamma(1+ n/2)}{n \cdot n!} Tr_\omega(a_0[\mathcal{D}, a_1]\cdots [\mathcal{D}, a_n]|\mathcal{D}|^{-n})  
 \end{equation}
\end{thm}
In the commutative triple $(C^\infty(M), L^2(M, \slashed{S}), \slashed{D})$, the Hochschild character $\Phi$ is computable 
by translating the Dixmier trace into a classical integral over $M$. 
In general, the Connes--Chern character is cohomologous to a $(b, B)$-cocycle defined by Wodzciki residue according to 
the noncommutative local index formula by Connes and Moscovici \cite{CM95}. 
Here we recall  the  local index theorem in 3 dimensions. 
\begin{thm}[\cite{CM95}]\label{locindthm}
If $(\mathcal{A}, \mathcal{H}, \mathcal{D})$ is a regular 3-summable spectral triple   
and $u \in U(\mathcal{A})$ is a unitary operator, let $\mathcal{F} = \mathcal{D} |\mathcal{D}|^{-1}$
be the sign of $\mathcal{D}$  and $P$ be the projection $({\mathcal{F}+1})/{2}$, then
the Fredholm index  can be computed by pairing $u \in K_1(\mathcal{A})$ with a $(b, B)$-cocycle $(\phi_1, \phi_3)$,
\begin{equation}
 Index(PuP) = \phi_1(u^*, u) - \phi_3(u^*, u, u^*, u)
\end{equation}
where 
\begin{equation}
 \begin{array}{ll}
    \phi_1(a^0,a^1) & =  \tau_0(a^0da^1|\mathcal{D}|^{-1})
                      -\frac{1}{4}\tau_0(a^0\nabla (da^1)|\mathcal{D}|^{-3}) \\
                     &  -\frac{1}{2}\tau_1(a^0\nabla (da^1)|\mathcal{D}|^{-3}) 
                       +\frac{1}{8}\tau_0(a^0\nabla^2 (da^1)|\mathcal{D}|^{-5}) \\
                     & +\frac{1}{3}\tau_1(a^0\nabla^2 (da^1)|\mathcal{D}|^{-5}) 
                       +\frac{1}{12}\tau_2(a^0\nabla^2 (da^1)|\mathcal{D}|^{-5})
 \end{array}
\end{equation}
and \begin{equation}
  \phi_3(a^0, a^1, a^2, a^3) =  \frac{1}{12} \tau_0(a^0da^1da^2da^3 |\mathcal{D}|^{-3}) 
                            + \frac{1}{6} \tau_1(a^0da^1da^2da^3|\mathcal{D}|^{-3})
\end{equation} with the notations  
\begin{equation*}
 \tau_k(a) =  Res_{z=0} z^k Tr(a|\mathcal{D}|^{-z}) 
\end{equation*}
$da = [\mathcal{D}, a]$ and $\nabla (a) = [\mathcal{D}^2, a]$.
\end{thm}

Let $\mathcal{B}$ be the algebra generated by the spaces 
$\delta^n(\mathcal{A})$ and $\delta^n([\mathcal{D}, \mathcal{A}])$
for all $n \geq 0$,
define a spectral zeta function for each $b \in \mathcal{B}$,
\begin{equation}
   \zeta_b(s) = Tr(b |\mathcal{D}|^{-s})
\end{equation}
which is analytic for $Re(s) \gg 0$. Then the dimension spectrum of a spectral triple 
is defined as the discrete singular points $\Sigma \subset \mathbb{C}$ of the  
meromorphic function $\zeta_b(s)$ after analytic continuation for all $b\in \mathcal{B}$.

When the dimension spectrum is simple, i.e., $\Sigma$ consists of only simple poles, then the $(b,B)$-cocycle $(\phi_1, \phi_3)$ 
can be simplified further as follows.
With the notation of noncommutative integral,
\begin{equation}
   \cutint a = Res_{z=0} \text{Tr}\, (a |\mathcal{D}|^{-z}) = \tau_0 (a)
\end{equation}
one now has
\begin{equation}
 \phi_1(a^0,a^1) = \cutint a^0da^1|\mathcal{D}|^{-1} - \frac{1}{4}\cutint a^0 \nabla(da^1) |\mathcal{D}|^{-3}+ \frac{1}{8}\cutint a^0 \nabla^2 (da^1)|\mathcal{D}|^{-5}
\end{equation}
\begin{equation}
 \phi_3(a^0, a^1, a^2, a^3) = \frac{1}{12}\cutint a^0 da^1da^2da^3 |\mathcal{D}|^{-3}
\end{equation}

In order to define the noncommutative Chern--Simons action, we first need to define connections, 
here we use the same definition as that in inner fluctuations of the spectral actoin \cite{CC06}.
Formally, one defines the space of 1-forms over $\mathcal{A}$ as the bimodule, 
 \begin{equation}
     \Omega^1_{\mathcal{D}}(\mathcal{A}) = \{\, \sum_i a_i [\mathcal{D}, b_i] \,; \quad a_i, b_i \in \mathcal{A} \,\}
 \end{equation}
then a connection 1-form  $A \in \Omega^1_{\mathcal{D}}(\mathcal{A})$ is a self-adjoint element, i.e. $A = A^*$.
\begin{defn}[\cite{P1201}]
  Let $(\mathcal{A}, \mathcal{H}, \mathcal{D})$ be a regular 3-summable spectral triple and  $A \in \Omega^1_{\mathcal{D}}(\mathcal{A})$ 
  be a connection 1-form,
the noncommutative Chern--Simons action is defined by
\begin{equation}
 S_{CS}(A) = 3 \phi_3(AdA + \frac{2}{3}A^3) - \phi_1(A)
\end{equation}
for the $(b, B)$-cocycle $(\phi_1, \phi_3)$ given in the 3-dimensional local index formula.
\end{defn}
When the linear term $\phi_1$ vanishes, this definition coincides with that of Connes and Chamseddine introduced in \cite{CC06}.
For example, $\phi_1$ does not vanish in the case of $SU_q(2)$, but it does vanish for $\mathbb{T}^3_\Theta$ or $S^3_\theta$.

\begin{thm}[\cite{P1201}]
The noncommutative Chern--Simons action is gauge invariant up to a Fredholm index,
\begin{equation}
 S_{CS}(u^*Au + u^*du) = S_{CS}(A) +  Index(PuP), \quad u \in U(\mathcal{A})
\end{equation}

\end{thm}
This can be verified directly by the properties of the $(b,B)$-cocycle  $(\phi_1, \phi_3)$ and the local index formula.
Once the gauge invariance of the Chern--Simons action is established, 
one can further study the noncommutative quantum Chern--Simons theory by a formal Feynman integral quantization.

\section{Quantum 3-sphere} \label{QuanSp}
In this section we first recall the Dirac geometry of the classical  three sphere  $S^3$, 
its Dirac spectrum can be computed from different approaches \cite{H74, H00, M11}. 
The quantum 3-sphere $S^3_\theta$ will be defined as a $\theta$-deformed $C^*$-algebra \cite{CL01}, and a spectral triple on $S^3_\theta$ 
will be constructed as a noncommutative analogue of the Dirac geometry of $S^3$. 

On the unit 3-sphere $S^3= \left\{(z_1,z_2)\in\mathbb{C}^2 : |z_1|^2 + |z_2|^2 = 1\right\} $, 
 the Hopf action is the isometric circle action,
\begin{equation*}
      S^1 \times S^3 \rightarrow S^3; \quad  (e^{i \omega}, (z_1,z_2) ) \mapsto (e^{i \omega} z_1, e^{i \omega}z_2)
\end{equation*}
or equivalently,  it is the matrix multiplication on $SU(2) \cong S^3$,
\begin{equation*}
      \begin{pmatrix} e^{i \omega} & 0 \\
                      0 & e^{-i \omega}
      \end{pmatrix}
      \begin{pmatrix} z_1 & z_2\\ 
                      -\bar{z}_2 & \bar{z}_1 
      \end{pmatrix} = 
      \begin{pmatrix} e^{i \omega}z_1 & e^{i \omega}z_2\\ 
                      -e^{-i \omega}\bar{z}_2 & e^{-i \omega}\bar{z}_1 
      \end{pmatrix} 
\end{equation*}
In addition, the Hopf map is defined by
\begin{equation*}
   h: S^3 \rightarrow S^2; \quad (z_1, z_2) \mapsto (|z_1|^2 - |z_2|^2, 2z_1 \bar{z}_2),
\end{equation*}
which induces the Hopf fibration $S^1 \hookrightarrow S^3 \xrightarrow{\ h \, } S^2$.

In real coordinates,
$$
S^3 =\{(x_0,x_1,x_2,x_3)\in\mathbb{R}^4 : x_0^2 + x_1^2 + x_2^2 + x_3^2 = 1 \}
$$
There exists a canonical choice of orthonormal right invariant vector fields in the tangent space 
$T_eS^3 \cong \mathfrak{su}(2) $ at $e =(1, 0, 0, 0)$, 
\begin{equation*}
  \begin{array}{ll}
          X =  -x_3 {\partial_0} - x_2 {\partial_1} 
          + x_1{\partial_2} +x_0{\partial_3} \\
    Y = - x_2{\partial_0} +x_3{\partial_1}  
          +x_0{\partial_2} - x_1{\partial_3} \\ 
   Z =  -x_1{\partial_0} + x_0{\partial_1}  
          -x_3{\partial_2} + x_2{\partial_3} 
     \end{array}
\end{equation*}
where $  Z = \partial_\omega$ can be identified with the velocity field of the rotation in the Hopf action. 
 The Dirac operator acting on the spinors $L^2(S^3, \slashed{S})$ in the left trivialization of the spin bundle $\slashed{S}$
 was defined in \cite{H00},
\begin{equation*}
 {D} = \frac{3}{2}I_2 +  iX\sigma_1 +   iY \sigma_2 + iZ \sigma_3
\end{equation*}
with Pauli matrices
\begin{equation*}
 \sigma_1  = \begin{pmatrix} 0&1\\ 1&0 \end{pmatrix}, \quad 
 \sigma_2  = \begin{pmatrix} 0&-i\\ i&0 \end{pmatrix}, \quad
 \sigma_3  = \begin{pmatrix} 1&0\\ 0&-1 \end{pmatrix} 
\end{equation*}
For convenience, sometimes the Dirac operator without the constant matrix $\begin{pmatrix}
                                                                            {3}/{2} & 0 \\
                                                                            0 & {3}/{2}
                                                                           \end{pmatrix}$ is denoted by $D'$, 
\begin{equation*}
 {D}' =  iX\sigma_1 +   iY \sigma_2 +  iZ\sigma_3
\end{equation*}

If one identifies  $z_1 = x_0+ ix_1,\, z_2 = x_2+ ix_3$, then in complex coordinates
\begin{equation*}
  \begin{array}{ll}
            X =  -i( \bar{z}_2 \partial_{z_1} - {z}_2 \partial_{\bar{z}_1} - \bar{z}_1 \partial_{z_2} + {z}_1 \partial_{\bar{z}_2})  \\
        Y = - ( \bar{z}_2 \partial_{z_1} + {z}_2 \partial_{\bar{z}_1} - \bar{z}_1 \partial_{z_2} - {z}_1 \partial_{\bar{z}_2} ) \\
            Z =    i ({z}_1 \partial_{z_1} - \bar{z}_1 \partial_{\bar{z}_1} + {z}_2 \partial_{z_2} - \bar{z}_2 \partial_{\bar{z}_2})
  \end{array}
\end{equation*}
It is convenient to define the ladder operators, 
\begin{equation*}
   \begin{array}{ll}
    L_+ = X - iY  =  2i    ({z}_2 \partial_{\bar{z}_1} - {z}_1 \partial_{\bar{z}_2} ) \\
    L_- = X + iY =  -2i (\bar{z}_2 \partial_{{z}_1} - \bar{z}_1 \partial_{{z}_2} ) 
    \end{array}
\end{equation*}
They satisfy the commutation relations
\begin{equation*}
       [Z , L_+] = 2 iL_+,\quad  [Z, L_-] = -2iL_-, \quad [L_+, L_-] = 4iZ    
\end{equation*}
In other words, $H =  -iZ/2, E =  iL_+ / 2\sqrt{2}, F = iL_- / 2\sqrt{2} $ give a representation of the Lie algebra $\mathfrak{su}(2)$, i.e.,
\begin{equation*}
  [H, E] = E, \quad [H, F] =- F, \quad [E, F] = H
\end{equation*}

In the Hopf coordinate system, whose geometric picture is the join operation $S^1 \star S^1 = S^3$,
the complex coordinates are expressed as
\begin{equation*}
    z_1 = e^{i\,\xi_1} \cos  \eta, \quad
    z_2 = e^{i\,\xi_2} \sin  \eta, \quad \xi_i \in [0, 2 \pi], \,\, \eta \in [0, \pi/2]   
\end{equation*}  
and the vector fields are written as
\begin{equation*}
  \begin{array}{ll}
         L_+ =  - {e^{i(\xi_1 + \xi_2)}} [   ( \tan \eta \partial_{\xi_1}  - \cot \eta \partial_{\xi_2} )  +i \partial_\eta]  \\
      L_-=  -{e^{-i(\xi_1 + \xi_2)}} [  (\tan \eta \partial_{\xi_1}  - \cot \eta \partial_{\xi_2}) - i \partial_\eta] \\
              Z = \partial_{\xi_1} + \partial_{\xi_2} 
  \end{array}
\end{equation*}

The Casimir operator for $\mathfrak{su}(2)$ is given by
\begin{equation*}
   \begin{array}{ll}
        C   &= H^2 + FE + EF = -\frac{1}{4} [Z^2 + (L_+L_- + L_-L_+)/2] \\
          & =  -\frac{1}{4} ( \partial_\eta^2 + \sec^{2} \eta \partial_{\xi_1}^2 +  \csc^{2}  \eta\partial_{\xi_2}^2  )
   \end{array}
\end{equation*}
so  the invariant Dirac Laplacian  is related to the Casimir operator by  
\begin{equation*}
   {{D}'}^2=  -\partial_\eta^2 - \sec^2 \eta \partial_{\xi_1}^2 - \csc^2 \eta \partial_{\xi_2}^2 = 4C
\end{equation*}

\begin{defn}
The Dirac operator on $S^3$ in the Hopf coordinates is defined by 
 \begin{equation}
   \slashed{D} =  \frac{3}{2}I_2 + i \begin{pmatrix}
                    Z & {L}_+ \\
                    {L}_- & -Z
                 \end{pmatrix} 
\end{equation}
where
\begin{equation*}
  \begin{array}{ll}
        {L}_+ =   -i {e^{i(\xi_1 + \xi_2)}} [  \partial_\eta -i ( \tan \eta \partial_{\xi_1}  - \cot \eta \partial_{\xi_2} )  ]  \\
      {L}_-=   i {e^{-i(\xi_1 + \xi_2)}} [    \partial_\eta +i (\tan \eta \partial_{\xi_1}  - \cot \eta \partial_{\xi_2}) ] \\
           Z = \partial_{\xi_1} + \partial_{\xi_2} 
  \end{array}
\end{equation*}
\end{defn}
The convenient notation for the Dirac operator without the constant term is also used as before,
\begin{equation}
   \slashed{D}' = i \begin{pmatrix}
                    Z & {L}_+ \\
                    {L}_- & -Z
                 \end{pmatrix} 
\end{equation}
then 
\begin{equation*}
   {\slashed{D}'}^2=  - \partial_\eta^2 - \sec^2 \eta \partial_{\xi_1}^2 - \csc^2  \eta \partial_{\xi_2}^2 
\end{equation*}
corresponds to the round metric on $S^3$,
\begin{equation*}
   ds^2 = d\eta^2 +  \cos^2 \eta d\xi_1^2 + \sin^2 \eta d \xi_2^2
\end{equation*}

By the Peter-Weyl theorem, one has an orthogonal Hilbert basis for $L^2(SU(2), d\mu)$ with $d\mu$  the standard Haar measure on $SU(2)$, 
\begin{equation*}
  \phi^m_{i, j}(g)= \binom{m}{i}^{-1/2} \binom{m}{j}^{-1/2} \sum_{s+t = i} 
                         \binom{m-j}{s}\binom{j}{t}  z_1^t (-\bar{z}_2)^{j-t}  z_2^s  \bar{z}_1^{m-j-s}
\end{equation*} where $m \geq 0, \,\,  0 \leq i,j \leq m$ such that 
\begin{equation} \label{orthbasis}
   \int_{SU(2)} \phi^m_{i, j}(g) \overline{\phi^n_{k, l}} (g)  \,d\mu(g)= \frac{1}{m+1} \delta_{mn}\delta_{ik}\delta_{jl}
 \end{equation}

Denote the coefficients by
\begin{equation*} 
    \begin{array}{ll}
         c_{i, j}^m =  \binom{m}{i}^{-1/2} \binom{m}{j}^{-1/2},\quad    b_{s, t}^{m,j} =   \binom{m-j}{s}\binom{j}{t} 
    \end{array}
\end{equation*}
in the Hopf coordinates,
\begin{equation*}
  \phi^m_{l, j}= c_{l, j}^m \sum_{s+t = l} (-1)^{j-t}
                          b_{s, t}^{m,j}  e^{i(l+j -m )\xi_1} e^{i(l - j)\xi_2} (\cos\eta)^{m-j -s +t} (\sin\eta)^{j -t +s}
\end{equation*}
It it straightforward to check that
\begin{equation*}
   \begin{array}{ll}
             Z\phi^m_{l,j} = i(2l-m) \phi^m_{l,j} \\
         {L}_+ \phi^m_{l,j}  = 2i \sqrt{l+1}\sqrt{m-l} ~\phi^m_{l+1,j} \\
         {L}_- \phi^m_{l,j}  = 2i\sqrt{l}\sqrt{m -l+1 } ~\phi^m_{l-1,j} \\

   \end{array}
\end{equation*}
and the Dirac Laplacian has eigenvalues $m(m+2)$  with multiplicity $(m+1)^2$,
\begin{equation*}
  {\slashed{D}'}^2 \phi^m_{l, j} = - [Z^2 + ({L}_+ {L}_- + {L}_- {L}_+)/2] \phi^m_{l, j}= (m^2+2m) \phi^m_{l, j}
\end{equation*}

One constructs the orthonormal eigenspinors in $L^2(S^3, \slashed{S})$ for the left trivialization as in \cite{H00},
\begin{equation*}
   \begin{array}{ll}
      \Phi^m_{k,\ell} = \begin{pmatrix}
                            
                               - \sqrt{k} \, \phi^m_{m-k+1,\ell} \\
                                 \sqrt{m-k+1}  \, \phi^m_{ m - k, \ell}  
                           \end{pmatrix} \quad (0 \leq k \leq m+1, \,  0 \leq \ell \leq m)
     \end{array}
\end{equation*}
\begin{equation*}
   \begin{array}{ll}
      \Phi^{-m}_{k,\ell} = \begin{pmatrix}
                                \sqrt{m-k+1} \, \phi^{m+1}_{m-k+1,\ell} \\
                               \sqrt{k+1}  \, \phi^{m+1}_{ m - k, \ell}  
                           \end{pmatrix} \quad (0 \leq k \leq m, \,  0 \leq \ell \leq m+1)
     \end{array}
\end{equation*}
Similarly one can  define eigenspinors based on left invariant vector fields and the right
trivialization of $L^2(S^3, \slashed{S})$.
It is easy to check that 
\begin{equation*}
  \begin{array}{ll}
      \slashed{D}' \Phi^m_{k,\ell}   =i \begin{pmatrix}
                            
                               Z  & {L}_+ \\
                                 {L}_-  &  -Z 
                           \end{pmatrix}
                           \begin{pmatrix}
                            
                               - \sqrt{k} \, \phi^m_{m-k+1,\ell} \\
                                 \sqrt{m-k+1}  \, \phi^m_{ m - k, \ell}  
                           \end{pmatrix}  =  m \, \Phi^m_{k,\ell}
      \end{array}
\end{equation*}
\begin{equation*}
  \begin{array}{ll}
      \slashed{D}' \Phi^{-m}_{k,\ell} = i \begin{pmatrix}                            
                               Z  & {L}_+ \\
                                 {L}_-  &  -Z  
                           \end{pmatrix}
                          \begin{pmatrix}
                                \sqrt{m-k+1} \, \phi^{m+1}_{m-k+1,\ell} \\
                               \sqrt{k+1}  \, \phi^{m+1}_{ m - k, \ell}  
                           \end{pmatrix}  = -(m+ 3) \, \Phi^{-m}_{k,\ell}
      \end{array}
\end{equation*}
Together with the Frobenius reciprocity, the space of spinors has a decomposition 
\begin{equation*}
   L^2(S^3, \slashed{S}) = H^- \oplus H^+ = \left( \oplus E_{-m}\right) \oplus \left( \oplus E_m \right)
\end{equation*}
where $E_m$ (resp. $E_{-m}$) is the eigenspace of $\slashed{D}$ with eigenvalue $m+3/2$ (resp. $-(m+3/2)$).
In addition, the multiplicity of the eigenvalues $\pm(m+3/2)$ is equal to the dimension of $E_{\pm m}$,
i.e. $dim E_{\pm m} = (m+1)(m+2)$. 
Our concrete construction is parallel to the representation theoretic approach in  \cite{H00}.

\begin{defn}
  The quantum 3-sphere $S^3_\theta$ is defined as the universal $C^*$-algebra generated by
operators $\alpha$ and $\beta$ satisfying the relations,
\begin{equation}
 \alpha \beta = \lambda \beta \alpha, \quad  \alpha^* \beta = \bar{\lambda} \beta \alpha^*, \quad 
 \alpha \alpha^* =  \alpha^* \alpha, \quad \beta \beta^* = \beta^* \beta, \quad \alpha \alpha^* + \beta \beta^* = 1
\end{equation}
for the complex parameter $\lambda = e^{2\pi i \theta}$ and irrational $\theta \in \mathbb{R} \setminus \mathbb{Q}$.
\end{defn}

In other words, $S^3_\theta$ is the $C^*$-algebraic version of
 the $\lambda$-deformed $SU(2)$, that is, 
\begin{equation*}
  w=\begin{pmatrix} \alpha & \beta \\
                      -\lambda \beta^* & \alpha^*
    \end{pmatrix} \in SU_\lambda(2) = S^3_\theta
\end{equation*} so that $w$ is a unitary operator.
$S^3_\theta$  was first introduced in \cite{CL01}, it is a special case of a more general class of noncommutative 3-spheres
considered  in \cite{CD02}. The K-groups of this quantum 3-sphere are simply given by,
\begin{equation*}
   K_0(S^3_\theta) \cong \mathbb{Z}, \quad   K_1(S^3_\theta) \cong \mathbb{Z}
\end{equation*}
It is also possible to generate $S^3_\theta$ by self-adjoint operators
and more details can be found in \cite{CD02}. 

There exists a natural parametrization of the generators in $S^3_\theta$ by Hopf coordinates,
\begin{equation}
 \alpha = u \cos \psi , \quad \beta = v  \sin \psi, \quad \psi \in [0, \pi/2]
\end{equation}
where $u,v$ are the generators of the noncommutative 2-torus $\mathbb{T}_\theta^2$ satisfying $uv = \lambda vu$.
One can define the Hopf circle action as usual 
and the Hopf map 
\begin{equation*}
   h: \quad (u \cos \psi, v \sin \psi ) \mapsto (\cos 2\psi, uv^* \sin  2 \psi)
\end{equation*}
gives rise to a quantum principal $U(1)$-Hopf fibration.

Over the quantum 3-sphere $S^3_\theta$, we define the Dirac operator as
\begin{equation}
 \mathcal{D}_1 =\frac{3}{2} I_2 + i \begin{pmatrix} 
  X_3  &  X^+ \\
  X^- & -X_3 
 \end{pmatrix} 
\end{equation}
where 
\begin{equation*}
  \begin{array}{ll}
     X_3 = i \delta_1 + i\delta_2 \\
    {X}^+ =-i {uv} [ \partial_\psi+ (\tan \psi \delta_1 - \cot \psi \delta_2  ) ]\\
    {X}^- = i {(uv)^*} [\partial_\psi - (\tan \psi \delta_1 - \cot \psi \delta_2)   ]  
   \end{array}
\end{equation*}
and $\delta_i$ are the canonical derivations on $\mathbb{T}^2_\theta$,
\begin{equation*}
 \delta_1(u) = u, \quad \delta_1(v) = 0, \quad \delta_2(u) = 0, \quad \delta_2(v) =v
\end{equation*}

In order to get the same Dirac spectrum we actually have to distinguish between left and right multiplications.
More precisely, let us use $L$ ( resp. $R$) to indicate the left (resp. right) multiplication, the ladder operators in the Dirac
operator should be defined as
\begin{equation}
  \begin{array}{ll}
     {X}^+ = -iL(u) R(v) [ \partial_\psi+ (\tan \psi \delta_1 - \cot \psi \delta_2  ) ]\\
    {X}^- =  iL(u^*) R(v^*) [\partial_\psi - (\tan \psi \delta_1 - \cot \psi \delta_2)   ]  
   \end{array}
\end{equation}
As expected, we have the same eigenvalues as before if these operators are applied  to 
\begin{equation*}
  \tilde{\phi}^m_{l, j}= c_{l, j}^m \sum_{s+t = l} 
                          b_{s, t}^{m,j} (-1)^{j-t} u^{l+j -m} v^{l - j} (\cos\psi)^{m-j -s +t} (\sin\psi)^{j -t +s}
\end{equation*}
The eigenspinors $\tilde{\Phi}^m_{l,j}$ can be defined similarly so that  $\mathcal{D}_1$ has the same Dirac spectrum as in 
the Dirac geometry of $S^3$. In other words, we have obtained the Hilbert space  of spinors, denoted by
$L^2(S^3_\theta, \mathbf{S}) $,  with complex coordinates replaced by  
the generators of $S^3_\theta$ in $\tilde{\Phi}^m_{l,j}$.

Denote by $C^\infty(S^3_\theta)$  the pre-$C^*$-algebra of smooth elements  $ a \in  C^\infty(S^3_\theta)$ of rapid decay, 
i.e.
\begin{equation*}
   \quad a =  \sum_{(k, m, n)} a_{kmn}\alpha^k \beta^m {\beta^*}^n, 
\end{equation*}
where $k \in \mathbb{Z}$ (so $\alpha^{-1}$ is understood as $\alpha^*$) and $m, n \in \mathbb{N}_0$ are non-negative integers, such that 
\begin{equation*}
 \quad  \{|k|^r m^s n^t |a_{kmn}|\}_{(k,m,n)\in \mathbb{Z} \times \mathbb{N}_0 \times \mathbb{N}_0} \subset B_d
\end{equation*}
i.e., the above sequence is bounded for any positive integer $r, s, t > 0$.

Putting together,  the spectral triple $(C^\infty(S^3_\theta), L^2(S^3_\theta, \mathbf{S}), \mathcal{D}_1 )$  
generalizes the Dirac geometry $(C^\infty(S^3), L^2(S^3, \slashed{S}), \slashed{D} )$. An alternative way to construct the same spectral triple is to introduce a Moyal product
into the commutative triple $(C^\infty(S^3), L^2(S^3, \slashed{S}), \slashed{D})$. More precisely,  
define a star product so that $(C^\infty(S^3), \star_\theta) = C^\infty(S^3_\theta)$, 
then the spectral triple consists of 
the same Dirac operator and Hilbert space but a new noncommutative smooth algebra $(C^\infty(S^3), \star_\theta)$.
More details about such Moyal star-product deformation can be found in the work by Rieffel \cite{R93}.

\section{Chern--Simons action} \label{CSaction}

We first check that the spectral triple $(C^\infty(S^3_\theta), L^2(S^3_\theta, \mathbf{S}), \mathcal{D}_1)$ 
satisfies the conditions of the
local index theorem and has simple dimension spectrum,  then we compute the Chern--Simons action in this section.

For later convenience,  we write the Dirac operator as 
\begin{equation*}
  \mathcal{D}_1 =  \frac{3}{2} I_2 + \begin{pmatrix}
  \slashed{\partial}_3  & \slashed{\partial}^+ \\
  \slashed{\partial}^- & -\slashed{\partial}_3 
 \end{pmatrix}   = \frac{3}{2} I_2 + \slashed{\partial}_1 \sigma_1 + \slashed{\partial}_2 \sigma_2 + \slashed{\partial}_3 \sigma_3 
\end{equation*}
where
\begin{equation*}
   \begin{array}{ll}
    \slashed{\partial}_3 =  - (\delta_1 + \delta_2) \\
            \slashed{\partial}^+ = L(u) R(v) [ \partial_\psi+ (\tan \psi \delta_1 - \cot \psi \delta_2) ]\\
        \slashed{\partial}^- =  - L(u^*) R(v^*) [ \partial_\psi - (\tan \psi \delta_1 - \cot \psi \delta_2   ) ]         
   \end{array}
\end{equation*}
and
\begin{equation*}
   \slashed{\partial}_1 = \frac{1}{2} (\slashed{\partial}^+ + \slashed{\partial}^-), \quad 
   \slashed{\partial}_2 = \frac{i}{2} (\slashed{\partial}^+ - \slashed{\partial}^-)
\end{equation*}

One can also express $\slashed{\partial}^+, \slashed{\partial}^-$ in terms of $\alpha, \beta$ and their adjoints as in the complex coordinates, the 
commutators between the Dirac operator and the generators are
\begin{equation*}
           [ \mathcal{D}_1, \alpha] =  \beta^* (\sigma_1 - i \sigma_2)  - \alpha    \sigma_3  = \begin{pmatrix}
                                                                                                  - \alpha & 0\\
                                                                                                  2\beta^* & \alpha
                                                                                                \end{pmatrix}
\end{equation*}
\begin{equation*}
             [\mathcal{D}_1 , \beta ] =   -\alpha^*   (\sigma_1 - i \sigma_2)  - \beta  \sigma_3 =  \begin{pmatrix}
                                                                                                  - \beta & 0\\
                                                                                                  -2\alpha^* & \beta
                                                                                                \end{pmatrix}
\end{equation*}
\begin{equation*}
            [\mathcal{D}_1 , \alpha^* ] =  -\beta (\sigma_1 + i \sigma_2)  + \alpha^*    \sigma_3   = \begin{pmatrix}
                                                                                                   \alpha^* & -2 \beta\\
                                                                                                  0 & -\alpha^*
                                                                                                \end{pmatrix}
\end{equation*}
\begin{equation*}
             [ \mathcal{D}_1, \beta ^*] =  \alpha (\sigma_1 + i \sigma_2) + \beta^*  \sigma_3  = \begin{pmatrix}
                                                                                                   \beta^* & 2 \alpha\\
                                                                                                  0 & -\beta^*
                                                                                                \end{pmatrix}
\end{equation*}
so the commutator $[\mathcal{D}_1, a]$ for any $a \in \mathcal{A} =C^\infty(S^3_\theta)$ is a bounded operator.
Furthermore, $(C^\infty(S^3_\theta), L^2(S^3_\theta, \mathbf{S}), \mathcal{D}_1)$ is a 3-summable spectral triple
since  the Dirac operator $\mathcal{D}_1$ has the same spectrum 
as in the Dirac geometry. 

For the pseudo-differential calculus, we use the conventional notations,
\begin{equation*}
   OP^{0} = \cap_{n = 1}^\infty Dom \delta^n, \quad OP^k = |\mathcal{D}|^k OP^0, \quad OP^{-\infty} = \cap_{k >0} OP^{-k}
\end{equation*}
As for the regularity condition,  i.e. $\mathcal{A} \subset OP^{0} $ and $ [\mathcal{D}_1, \mathcal{A}] \subset OP^{0}$, 
it is enough to check it on the generators of  $C^\infty(S^3_\theta)$.
Let  $\mathcal{F} = \mathcal{D}_1|\mathcal{D}_1|^{-1}$ be the sign of $\mathcal{D}_1$,
for example, $\delta(a) = [|\mathcal{D}_1|, a] = \mathcal{F}[\mathcal{D}_1, a] +  [\mathcal{F}, a]\mathcal{D}_1 $,
and $[\mathcal{F}, a]$ belongs to the two sided ideal $OP^{-\infty} \subset OP^{0}$,
similarly for $\delta([\mathcal{D}_1, a])  = [\mathcal{F}\mathcal{D}_1, [\mathcal{D}_1, a]]$. 
Together with the results of commutators with the generators, it is obvious that the spectral triple $(C^\infty(S^3_\theta), L^2(S^3_\theta, \mathbf{S}), \mathcal{D}_1)$  is regular.

Recall that $\mathcal{B}$ is the algebra generated by $\delta^n(\mathcal{A})$ and $\delta^n([\mathcal{D}_1, \mathcal{A}])$ for all $n \geq 0$, 
here $\mathcal{A}= C^\infty(S^3_\theta)$.
For each $b \in \mathcal{B}$, the zeta function $\zeta_b (z) = Tr(b |\mathcal{D}_1|^{-z})$ is analytic for $Re(z) > 3$, let us check the 
dimension spectrum of the spectral triple $(C^\infty(S^3_\theta), L^2(S^3_\theta, \mathbf{S}), \mathcal{D}_1)$. Using the orthonormal basis of 
$L^2(S^3_\theta, \mathbf{S})$, the trace can be expressed explicitly as,
\begin{equation*} 
   \begin{array}{ll}
       & \zeta_b (z)  =Tr(b |\mathcal{D}_1|^{-z}) \\
       & = \sum_{m \geq 0} \sum_{k, \ell} <\tilde{\Phi}^m_{k,\ell}, b |\mathcal{D}_1|^{-z} \tilde{\Phi}^m_{k,\ell} >
                                    + <\tilde{\Phi}^{-m}_{k,\ell}, b |\mathcal{D}_1|^{-z} \tilde{\Phi}^{-m}_{k,\ell} >\\
                                 & = \sum_{m \geq 0} \sum_{k, \ell} (m+3/2)^{-z} [<\tilde{\Phi}^{m}_{k,\ell}, b \tilde{\Phi}^{m}_{k,\ell} > 
                                                    + <\tilde{\Phi}^{- m}_{k,\ell}, b \tilde{\Phi}^{-m}_{k,\ell} >]
   \end{array}
\end{equation*}
In general, $b \in \mathcal{B}$ is a   $2 \times 2$ matrix with entries being functions in the generators of $S^3_\theta$.
So this reduces the problem to consider $ <\tilde{\phi}^m_{k,\ell}, O\tilde{\phi}^m_{k,\ell}> $ for an arbitrary operator valued function $O(\alpha, \beta)$, 
but only the constant term contributes, i.e., for some $b_0 \in \mathbb{C}$, 
\begin{equation*}
  \begin{array}{ll}
        Tr(b |\mathcal{D}_1|^{-z}) & = \sum_{m \geq 0} (m+1)(m+2) (m+3/2)^{-z}  b_0 \\
                                   & = b_0 [\zeta_H(z-2, 3/2) - \frac{1}{4} \zeta_H(z, 3/2) ]

  \end{array}
\end{equation*}
where $\zeta_H(z, a)$ is the Hurwitz zeta function. Since $\zeta_H(z, a)$ only has a simple pole at $z =1$, 
the spectral triple $(C^\infty(S^3_\theta), L^2(S^3_\theta, \mathbf{S}), \mathcal{D}_1)$ 
has simple dimension spectrum $\{1, 3 \}$.

\begin{prop}
 The first cochain $\phi_1$ in the Chern--Simons action vanishes  for the spectral triple  
 $(C^\infty(S^3_\theta), L^2(S^3_\theta, \mathbf{S}), \mathcal{D}_1)$ if the Dirac Laplacian is used in
 $\nabla(a) = [\mathcal{D}_1'^2, a]$.
\end{prop} 

\begin{proof}
Recall that for any $a^0, a^1 \in C^\infty(S^3_\theta)$,
\begin{equation*}
 \phi_1(a^0,a^1) = \cutint a^0da^1|\mathcal{D}_1|^{-1} - \frac{1}{4}\cutint a^0 \nabla(da^1) |\mathcal{D}_1|^{-3}+ \frac{1}{8}\cutint a^0 \nabla^2 (da^1)|\mathcal{D}_1|^{-5}
\end{equation*}
We have $d a^1 = [\mathcal{D}_1, a^1 ] =  [\slashed{\partial}_k, a^1]    \sigma^k$, 
and each term has a Pauli matrix whose trace is zero, so the first noncommutative integral vanishes. 

Next we consider $\nabla(da^1)= [\mathcal{D}_1'^2, da^1] $ and  $\nabla^2(da^1)= [\mathcal{D}_1'^2, \nabla(da^1)] $. 
As a convention, the Dirac Laplacian is used, and the relation is clear 
$\mathcal{D}^2 = (3/2I_2 + \mathcal{D}')^2 = {9}/{4}I_2 + 3\mathcal{D}' + \mathcal{D}'^2$. 
Another reason to use the Dirac Laplacian is to compare with other spectral triples in the next section.
Since $  {\mathcal{D}_1'}^2=   (\sec^2 \psi \delta_1^2 + \csc^2 \psi \delta_2^2 -\partial_\psi^2 )I_2$,
 $\nabla(da^1)$  and $\nabla^2(da^1)$ still have Pauli matrices in each term, so the other two noncommutative integrals also vanish.

\end{proof}

Since the linear term disappears  on $S^3_\theta$, the noncommutative Chern--Simons action
is a direct generalization of the classical Chern--Simons action over the 3-sphere. 

Any principal bundle over $S^3$ is trivializable, 
so we assume a connection 1-form over $S_\theta^3$  is a self-adjoint
element in the bimodule $\Omega^1_{\mathcal{D}_1}(S_\theta^3)$,
\begin{equation}
  A = \sum_i a_i [\mathcal{D}_1, b_i] = \sum_i a_i [\slashed{\partial}_k, b_i] \sigma^k = A_k \sigma^k, \quad a_i, b_i \in C^\infty(S_\theta^3)
\end{equation}
Since $tr(\sigma_i \sigma_j \sigma_k) = 2i \varepsilon^{ijk} $,
the Chern--Simons action on $S^3_\theta$ is 
\begin{equation*}
   \begin{array}{ll}
    S_{CS}(A) & =  \phi_3(  3 A \wedge [\mathcal{D}_1, A] + 2 A \wedge A \wedge A ) \\
              & =   \phi_3[\sigma_i \sigma_j \sigma_k  (3 A_i [ \slashed{\partial}_j, A_k]  + 2 A_i  A_j   A_k)]  \\
              & = ({i }/{6} )  \fint \, \varepsilon^{ijk} ( 3 A_i [\slashed{\partial}_j, A_k]   + 2 A_iA_j A_k)|\mathcal{D}_1|^{-3} \\
              & = ({i}/{6} )  Res_{z=0}Tr \, \varepsilon^{ijk} (3 A_i [\slashed{\partial}_j, A_k] +  {2}A_i A_j   A_k)|\mathcal{D}_1|^{-3-z}
   \end{array}
\end{equation*}

\begin{thm}\label{NCCSaction}
   The Chern--Simons action on the quantum 3-sphere $S^3_\theta$ 
   with respect to the spectral triple $(C^\infty(S^3_\theta), L^2(S^3_\theta, \mathbf{S}), \mathcal{D}_1 )$  is given by   
\begin{equation} 
    \begin{array}{ll}
     &S_{CS}(A) 
      =    2 \sum [( - 2n'k -2k'n +  k'k(q+n))  \lambda -  k'k(q+n+1)]  \\
    & [{a'}_{q'q'k'}b_{k'n'n'} {a'}_{qqk} {b}_{knn} + a_{q'q'k'}{b'}_{k'n'n'} {a}_{qqk} {b'}_{knn}]+  \\
     & [- 2n'k -2k'n  - k'k(q+n))  \lambda + k'k(q+n+1)] \\
     &     [{a}_{q'q'k'}b'_{k'n'n'} {a'}_{qqk} {b}_{knn} + a'_{q'q'k'}{b}_{k'n'n'} {a}_{qqk} {b'}_{knn}] 
    \end{array}
\end{equation}
for $A = a [\mathcal{D}_1, b]$, $a, b \in C^\infty(S^3_\theta)$. 
\end{thm}

\begin{rmk}
   In particular, if $a$ and $b$ are of the special form so that their coefficients are symmetric, i.e., $a_{qqk} = a'_{qqk}$ and $b_{knn} = b'_{knn}$, 
   then the above Chern--Simons action can be reduced to 
   $$
     S_{CS}(A) 
      =    -8 \lambda  \sum ( n'k+k'n)       {a}_{q'q'k'}b_{k'n'n'} {a}_{qqk} {b}_{knn}  
   $$
\end{rmk}

\begin{proof}
 First we assume that $A = a [\mathcal{D}_1, b] = a [\slashed{\partial}_k, b] \sigma^k$ has only one generic term,
\begin{equation*}
  \begin{array}{ll}
     a   = \sum_{(p, q, \ell) } \mathbf{a}^+_{pq\ell} +  \mathbf{a}^-_{pq\ell} =\sum_{(p, q, \ell) }  a_{pq\ell} \beta^p {\beta^*}^q \alpha^\ell + a'_{pq\ell} \beta^p {\beta^*}^q {\alpha^*}^\ell, \\
     b   = \sum_{(k,m,n) } \mathbf{b}^+_{kmn} + \mathbf{b}^-_{kmn} =\sum_{(k,m,n) }b_{kmn}\alpha^k \beta^m {\beta^*}^n + b'_{kmn}{\alpha^*}^k \beta^m {\beta^*}^n
  \end{array}
\end{equation*}
where $a_{pq\ell}$, $a'_{pq\ell}$, $b_{kmn}$,  $b'_{kmn}$ are coefficients of rapid decay.
For simplicity, we also assume that $p,q ,\ell$ and $k,m,n$ are all positive integers since the constant terms 
such as $a_{000}$ and $b_{000}$ can be recovered easily.
In addition, we expect the powers of $\alpha$ and $\beta$ will be canceled out after taking the trace, 
so $a$ is written as above to cancel $\alpha$ easily with that from $b$ without generating redundant $\lambda$ because of $uv = \lambda vu$.

The diagonal region $R$ is defined to be the connections $A$ whose coefficients 
are $a_{qqk}$, $a'_{qqk}$, $b_{knn}$,  $b'_{knn}$ etc. with identified powers of $\alpha$ and $\beta$, i.e., $p= q, \ell = k, m= n$.
The diagonal region $R$ is introduced to simplify the computation due to the orthogonal basis of $SU_\lambda(2)$,
see the classical case in \eqref{orthbasis}.

The components $A_i = a[\slashed{\partial}_i, b ]$ ($i = 1, 2, 3$) of the connection $A$ are
\begin{equation*}
   \begin{array}{ll}
       A_1  & =  \sum_{(k,m,n)}n\cot \psi \, au\mathbf{b}^+_{kmn}v + (k \tan \psi - m \cot \psi) au^*\mathbf{b}^+_{kmn}v^* \\
                              &+ (n \cot \psi - k \tan \psi)  au\mathbf{b}^-_{kmn}v  - m\cot \psi au^*\mathbf{b}^-_{kmn}v^*  \\  
   \end{array}
\end{equation*}
\begin{equation*}
  \begin{array}{ll}
    A_2 &  = i\sum_{(k,m,n)}n\cot \psi \, au\mathbf{b}^+_{kmn}v + ( m \cot \psi - k \tan \psi ) au^*\mathbf{b}^+_{kmn}v^* \\
                       &+ (n \cot \psi - k \tan \psi)  au\mathbf{b}^-_{kmn}v  + m\cot \psi au^*\mathbf{b}^-_{kmn}v^*  \\  
  \end{array}
\end{equation*}
\begin{equation*}
  \begin{array}{ll}
   A_3 & =  - \sum_{(k,m,n)} (k + m -n)a\mathbf{b}^+_{kmn} + (-k +m -n) a \mathbf{b}^-_{kmn}
  \end{array}
\end{equation*}
 Next, we compute the terms $A_i[\slashed{\partial}_j, A_k]$ for permutations $(i,j,k) \in S_3$, 
\begin{equation*}
   \begin{array}{ll}
       & A_1 [ \slashed{\partial}_2, A_3] \\
       & = [\sum_{(k',m',n')}n'\cot \psi \, au\mathbf{b}^+_{kmn}v + (k' \tan \psi - m' \cot \psi) au^*\mathbf{b}^+_{kmn}v^* \\
                              &+ (n' \cot \psi - k' \tan \psi)  au\mathbf{b}^-_{kmn}v  - m'\cot \psi au^*\mathbf{b}^-_{kmn}v^* ] \\  
                              & \{ -i \sum_{(k,m,n, p, q, \ell)} (k +m -n) [ (q +n) \cot \psi u\mathbf{a}^+_{pq\ell}\mathbf{b}^+_{kmn}v  \\
                               & + ((q +n)\cot \psi - \ell \tan \psi ) u\mathbf{a}^-_{pq\ell}\mathbf{b}^+_{kmn}v \\
                               & + ((q +n)\cot \psi - (\ell+k) \tan \psi ) u^*\mathbf{a}^+_{pq\ell}\mathbf{b}^+_{kmn}v^* \\
                                &+ ((q +n)\cot \psi - k \tan \psi ) u^*\mathbf{a}^-_{pq\ell}\mathbf{b}^+_{kmn}v^*] \\
                              &  -i \sum_{(k,m,n, p, q, \ell)} (-k +m -n) [ (q +n)\cot \psi  u^*\mathbf{a}^-_{pq\ell}\mathbf{b}^-_{kmn}v^* \\
                              & + ((q +n)\cot \psi - (\ell+k) \tan \psi ) u\mathbf{a}^-_{pq\ell}\mathbf{b}^-_{kmn}v \\
                               & + ((q +n)\cot \psi - \ell \tan \psi ) u^*\mathbf{a}^+_{pq\ell}\mathbf{b}^-_{kmn}v^* \\
                                &+((q +n) \cot \psi - k \tan \psi) u\mathbf{a}^+_{pq\ell}\mathbf{b}^-_{kmn}v  ] 
                              \}
  \end{array}
 \end{equation*} 
 
 \begin{equation*}
   \begin{array}{ll}
        &  A_2  [ \slashed{\partial}_1, A_3] \\ 
        & = [ i \sum_{(k',m',n')}n'\cot \psi \, au\mathbf{b}^+_{kmn}v + ( m' \cot \psi - k' \tan \psi ) au^*\mathbf{b}^+_{kmn}v^* \\
                              &+ (n' \cot \psi - k' \tan \psi)  au\mathbf{b}^-_{kmn}v  + m'\cot \psi au^*\mathbf{b}^-_{kmn}v^* ] \\ 
                                & \{  \sum_{(k,m,n, p, q, \ell)} (k +m -n) [ -(q +n) \cot \psi u\mathbf{a}^+_{pq\ell}\mathbf{b}^+_{kmn}v  \\
                               & - ((q +n)\cot \psi - \ell \tan \psi ) u\mathbf{a}^-_{pq\ell}\mathbf{b}^+_{kmn}v \\
                               & + ((q +n)\cot \psi - (\ell+k) \tan \psi ) u^*\mathbf{a}^+_{pq\ell}\mathbf{b}^+_{kmn}v^* \\
                                &+ ((q +n)\cot \psi - k \tan \psi ) u^*\mathbf{a}^-_{pq\ell}\mathbf{b}^+_{kmn}v^*] \\
                              &  + \sum_{(k,m,n, p, q, \ell)} (-k +m -n) [ (q +n)\cot \psi  u^*\mathbf{a}^-_{pq\ell}\mathbf{b}^-_{kmn}v^* \\
                              & - ((q +n)\cot \psi - (\ell+k) \tan \psi ) u\mathbf{a}^-_{pq\ell}\mathbf{b}^-_{kmn}v \\
                               & + ((q +n)\cot \psi - \ell \tan \psi ) u^*\mathbf{a}^+_{pq\ell}\mathbf{b}^-_{kmn}v^* \\
                                &-((q +n) \cot \psi - k \tan \psi) u\mathbf{a}^+_{pq\ell}\mathbf{b}^-_{kmn}v  ] 
                              \}
  \end{array}
 \end{equation*}
 
 \begin{equation*}
   \begin{array}{ll}
        & A_1 [  \slashed{\partial}_3, A_2] \\
        & = [\sum_{(k',m',n')}n'\cot \psi \, au\mathbf{b}^+_{kmn}v + (k' \tan \psi - m' \cot \psi) au^*\mathbf{b}^+_{kmn}v^* \\
                              &+ (n' \cot \psi - k' \tan \psi)  au\mathbf{b}^-_{kmn}v  - m'\cot \psi au^*\mathbf{b}^-_{kmn}v^* ] \\ 
                              & \{ -i \sum_{(k,m,n, p, q, \ell)}  n \cot \psi (p-q +m -n +\ell +k +2) \mathbf{a}^+_{pq\ell} u \mathbf{b}^+_{kmn}v \\
                              & + n \cot \psi (p-q +m -n -\ell +k +2) \mathbf{a}^-_{pq\ell} u \mathbf{b}^+_{kmn}v \\
                              & + (m \cot \psi - k \tan \psi )(p-q +m -n +\ell +k -2)\mathbf{a}^+_{pq\ell} u^* \mathbf{b}^+_{kmn}v^* \\
                              & + (m \cot \psi - k \tan \psi )(p-q +m -n -\ell +k -2)\mathbf{a}^-_{pq\ell} u^* \mathbf{b}^+_{kmn}v^* \\
                               & + (n \cot \psi - k \tan \psi )(p-q +m -n +\ell -k +2)\mathbf{a}^+_{pq\ell} u \mathbf{b}^-_{kmn}v \\
                              & + (n \cot \psi - k \tan \psi )(p-q +m -n -\ell -k +2)\mathbf{a}^-_{pq\ell} u \mathbf{b}^-_{kmn}v \\
                               & + m \cot \psi (p-q +m -n +\ell -k -2)\mathbf{a}^+_{pq\ell} u^* \mathbf{b}^-_{kmn}v^* \\
                              & + m \cot \psi (p-q +m -n -\ell -k -2)\mathbf{a}^-_{pq\ell} u^* \mathbf{b}^-_{kmn}v^* 
                              \}
  \end{array}
 \end{equation*}

\begin{equation*}
   \begin{array}{ll}
          & A_3  [  \slashed{\partial}_1, A_2] \\
          & = [- \sum_{(k', m',n')} (k'+m' -n')a\mathbf{b}^+_{k'm'n'} + (-k'+m'-n')a\mathbf{b}^-_{k'm'n'}] \\
                                              & \{ i \sum_{(k,m,n, p, q, \ell)} n(q+n-1) \cot^2 \psi uaubv^2 - k(q+n)uau\mathbf{b}^-_{kmn}v^2 \\
                                              &   + m (q+n) \cot^2 \psi uau^*b - k(q+n+1) uau^*\mathbf{b}^+_{kmn} \\
                                              & - m(\ell +k +1) u\mathbf{a}^-_{pq\ell}u^*\mathbf{b}^-_{kmn} + k(\ell +k -1) \tan^2 \psi u\mathbf{a}^-_{pq\ell}u\mathbf{b}^-_{kmn}v^2 \\
                                              & + k(p+m)u^*au^* \mathbf{b}^+_{kmn}{v^*}^2 + k(p+m +1) u^*au\mathbf{b}^-_{kmn} \\
                                              & - n(p+m) \cot^2 \psi u^*aub - m(p+m-1)\cot^2 \psi u^* au^*b {v^*}^2 \\
                                              & +n(\ell +k +1) u^*\mathbf{a}^+_{pq\ell}u\mathbf{b}^+_{kmn} - k(\ell +k -1)\tan^2 \psi u^* \mathbf{a}^+_{pq\ell}u^*\mathbf{b}^+_{kmn}{v^*}^2 \\
                                              & -\frac{1}{2}n(\ell +k +1) [u\mathbf{a}^+_{pq\ell}u\mathbf{b}^-_{kmn}v^2 + u\mathbf{a}^-_{pq\ell}u {b} v^2 ] \\
                                              & - \frac{1}{2}m(\ell + k + 1) [u\mathbf{a}^+_{pq\ell}u^* {b} + u\mathbf{a}^-_{pq\ell}u^* \mathbf{b}^+_{kmn} ] \\
                                              & + \frac{1}{2}k(\ell + k -1) \tan^2 \psi[u\mathbf{a}^-_{pq\ell}u^* \mathbf{b}^+_{kmn} + u\mathbf{a}^+_{pq\ell}u \mathbf{b}^-_{kmn} v^2 ] \\
                                              & +\frac{1}{2}n(-\ell + k +1) u\mathbf{a}^-_{pq\ell}u \mathbf{b}^+_{kmn} v^2 +\frac{1}{2}n(-\ell - k +1) u\mathbf{a}^-_{pq\ell}u \mathbf{b}^-_{kmn} v^2 \\
                                              &  +\frac{1}{2}m(\ell + k - 1) u\mathbf{a}^+_{pq\ell}u^* \mathbf{b}^+_{kmn} +  \frac{1}{2}m(\ell - k - 1) u\mathbf{a}^+_{pq\ell}u^* \mathbf{b}^-_{kmn} \\
                                              & + \frac{1}{2}[m(-\ell +k -1) - k (-\ell + k -1)\tan^2 \psi] u\mathbf{a}^-_{pq\ell}u^* \mathbf{b}^+_{kmn} \\
                                              & + \frac{1}{2}[n(\ell -k +1) - k (\ell - k +1)\tan^2 \psi] u\mathbf{a}^+_{pq\ell}u \mathbf{b}^-_{kmn} v^2\\
                                              &  +\frac{1}{2}n(\ell + k +1) [u^*\mathbf{a}^+_{pq\ell}u \mathbf{b}^-_{kmn} + u^*\mathbf{a}^-_{pq\ell}u {b} ] \\ 
                                              & +\frac{1}{2}m(\ell + k + 1) [u^*\mathbf{a}^+_{pq\ell}u^* {b}{v^*}^2 + u^*\mathbf{a}^-_{pq\ell}u^* \mathbf{b}^+_{kmn}{v^*}^2 ] \\
                                              & - \frac{1}{2}k(\ell+ k -1) \tan^2 \psi [u^*\mathbf{a}^-_{pq\ell}u^* \mathbf{b}^+_{kmn}{v^*}^2 + u^*\mathbf{a}^+_{pq\ell}u \mathbf{b}^-_{kmn} ] \\
                                              & + \frac{1}{2}n(-\ell + k + 1) u^*\mathbf{a}^-_{pq\ell}u \mathbf{b}^+_{kmn}  + \frac{1}{2}m(\ell + k -1) u^*\mathbf{a}^+_{pq\ell}u^* \mathbf{b}^+_{kmn}{v^*}^2 \\
                                              & + \frac{1}{2}n(-\ell - k + 1) u^*\mathbf{a}^-_{pq\ell}u \mathbf{b}^-_{kmn}  + \frac{1}{2}m(\ell - k -1) u^*\mathbf{a}^+_{pq\ell}u^* \mathbf{b}^-_{kmn}{v^*}^2 \\
                                              & + \frac{1}{2}[m (-\ell + k -1) - k(-\ell + k -1) \tan^2 \psi] u^*\mathbf{a}^-_{pq\ell}u^* \mathbf{b}^+_{kmn}{v^*}^2 \\
                                              & + \frac{1}{2}[n(\ell - k  +1) - k (\ell - k + 1) \tan^2 \psi] u^*\mathbf{a}^+_{pq\ell}u \mathbf{b}^-_{kmn}
                                             
  \}
                
  \end{array}
 \end{equation*}
 
 \begin{equation*}
   \begin{array}{ll}
       & A_2  [ \slashed{\partial}_3, A_1]  \\
       &=  [i\sum_{(k',m',n')}n'\cot \psi \, au\mathbf{b}^+_{kmn}v + ( m' \cot \psi - k' \tan \psi ) au^*\mathbf{b}^+_{kmn}v^* \\
                              &+ (n' \cot \psi - k' \tan \psi)  au\mathbf{b}^-_{kmn}v  + m'\cot \psi au^*\mathbf{b}^-_{kmn}v^* ] \\ 
                              & \{ - \sum_{(k,m,n, p, q, \ell)}  n \cot \psi (p-q +m -n +\ell +k +2) \mathbf{a}^+_{pq\ell} u \mathbf{b}^+_{kmn}v \\
                              & + n \cot \psi (p-q +m -n -\ell +k +2) \mathbf{a}^-_{pq\ell} u \mathbf{b}^+_{kmn}v \\
                              & + (k \tan \psi - m \cot \psi  )(p-q +m -n +\ell +k -2)\mathbf{a}^+_{pq\ell} u^* \mathbf{b}^+_{kmn}v^* \\
                              & + (k \tan \psi - m \cot \psi  )(p-q +m -n -\ell +k -2)\mathbf{a}^-_{pq\ell} u^* \mathbf{b}^+_{kmn}v^* \\
                               & + (n \cot \psi - k \tan \psi )(p-q +m -n +\ell -k +2)\mathbf{a}^+_{pq\ell} u \mathbf{b}^-_{kmn}v \\
                              & + (n \cot \psi - k \tan \psi )(p-q +m -n -\ell -k +2)\mathbf{a}^-_{pq\ell} u \mathbf{b}^-_{kmn}v \\
                               & - m \cot \psi (p-q +m -n +\ell -k -2)\mathbf{a}^+_{pq\ell} u^* \mathbf{b}^-_{kmn}v^* \\
                              & - m \cot \psi (p-q +m -n -\ell -k -2)\mathbf{a}^-_{pq\ell} u^* \mathbf{b}^-_{kmn}v^* 
                              \}
  \end{array}
 \end{equation*}
 
\begin{equation*}
   \begin{array}{ll}
      & A_3  [ \slashed{\partial}_2, A_1] \\
      & = [- \sum_{(k', m',n')} (k'+m' -n')a\mathbf{b}^+_{k'm'n'} + (-k'+m'-n')a\mathbf{b}^-_{k'm'n'}] \\
                                              & \{ i \sum_{(k,m,n, p, q, \ell)} n(q+n-1) \cot^2 \psi uaubv^2  + n (p+m) \cot^2 \psi u^*aub  \\
                                              &   + k(p+m)u^*au^*\mathbf{b}^+_{kmn}{v^*}^2 + k(q+n+1) uau^*\mathbf{b}^+_{kmn} \\
                                              & - k(p +m +1) u^*{a}u\mathbf{b}^-_{kmn} - k(q + n)  u{a}u\mathbf{b}^-_{kmn}v^2 \\
                                              & - m \cot^2 \psi (q +n)uau^* {b} -m\cot^2 \psi (p+m -1) u^*au^*{b}{v^*}^2 \\
                                              & +m(\ell + k +1) u\mathbf{a}^-_{pq\ell} u^* \mathbf{b}^-_{kmn} - n(\ell +k +1)  u^* \mathbf{a}^+_{pq\ell}u\mathbf{b}^+_{kmn}  \\
                                              & +k\tan^2 \psi(\ell +k -1) (u\mathbf{a}^-_{pq\ell}u\mathbf{b}^-_{kmn}v^2 - u^* \mathbf{a}^+_{pq\ell}u^*\mathbf{b}^+_{kmn}{v^*}^2) \\
                                              & -\frac{1}{2}n(\ell +k +1) [u\mathbf{a}^+_{pq\ell}u\mathbf{b}^-_{kmn}v^2 + u\mathbf{a}^-_{pq\ell}u {b} v^2 + u^*\mathbf{a}^+_{pq\ell}u \mathbf{b}^-_{kmn} + u^*\mathbf{a}^-_{pq\ell}ub ] \\
                                              & + \frac{1}{2}k\tan^2 \psi(\ell + k - 1) [u\mathbf{a}^+_{pq\ell}u \mathbf{b}^-_{kmn}v^2 - u\mathbf{a}^-_{pq\ell}u^* \mathbf{b}^+_{kmn} ] \\
                                              & + \frac{1}{2}k\tan^2 \psi(\ell + k - 1) [u^*\mathbf{a}^+_{pq\ell}u \mathbf{b}^-_{kmn} - u^*\mathbf{a}^-_{pq\ell}u^* \mathbf{b}^+_{kmn} {v^*}^2] \\
                                              & + \frac{1}{2}m(\ell + k +1) [u\mathbf{a}^-_{pq\ell}u^* \mathbf{b}^+_{kmn} + u\mathbf{a}^+_{pq\ell}u^* {b} ] \\
                                              &  + \frac{1}{2}m(\ell + k +1) [  u^*\mathbf{a}^-_{pq\ell}u^* \mathbf{b}^+_{kmn} {v^*}^2 +  u^*\mathbf{a}^+_{pq\ell}u^* {b} {v^*}^2] \\
                                              & +\frac{1}{2}n(-\ell + k +1) [ u\mathbf{a}^-_{pq\ell}u \mathbf{b}^+_{kmn} v^2 - u^*\mathbf{a}^-_{pq\ell}u \mathbf{b}^+_{kmn} ] \\
                                              &  +\frac{1}{2}m(\ell + k - 1) [u^*\mathbf{a}^+_{pq\ell}u^* \mathbf{b}^+_{kmn} {v^*}^2 - u\mathbf{a}^+_{pq\ell}u^* \mathbf{b}^+_{kmn} ]\\
                                              & + \frac{1}{2} k (-\ell + k -1)\tan^2 \psi[ u\mathbf{a}^-_{pq\ell}u^* \mathbf{b}^+_{kmn} - u^*\mathbf{a}^-_{pq\ell}u^* \mathbf{b}^+_{kmn}{v^*}^2 ]\\
                                              & + \frac{1}{2}m(-\ell +k -1) [ u^*\mathbf{a}^-_{pq\ell}u^* \mathbf{b}^+_{kmn} {v^*}^2 - u \mathbf{a}^-_{pq\ell} u^* \mathbf{b}^+_{kmn} ]\\
                                              &  +\frac{1}{2}n(\ell - k +1) [u\mathbf{a}^+_{pq\ell}u \mathbf{b}^-_{kmn}v^2 - u^*\mathbf{a}^+_{pq\ell}u \mathbf{b}^-_{kmn} ] \\ 
                                              & + \frac{1}{2}k(\ell- k +1) \tan^2 \psi [u^*\mathbf{a}^+_{pq\ell}u \mathbf{b}^-_{kmn} - u\mathbf{a}^+_{pq\ell}u \mathbf{b}^-_{kmn} v^2] \\
                                              & + \frac{1}{2}n(-\ell - k + 1) [u\mathbf{a}^-_{pq\ell}u \mathbf{b}^-_{kmn}v^2 -  u^*\mathbf{a}^-_{pq\ell}u \mathbf{b}^-_{kmn}] \\
                                              & + \frac{1}{2}m (\ell - k -1) [u^*\mathbf{a}^+_{pq\ell}u^* \mathbf{b}^-_{kmn}{v^*}^2 - u\mathbf{a}^+_{pq\ell}u^*\mathbf{b}^-_{kmn} ] 
                                                                            \}
  \end{array}
 \end{equation*}  
Combine the above terms together, we simplify them by setting $m = n, p = q, \ell = k$ and $m'=n'$, that is, restricting to the diagonal region, 
\begin{equation*}
   \begin{array}{ll}
      & A_1  [ \slashed{\partial}_2, A_3 ] - A_2 [\slashed{\partial}_1, A_3]|_{m'=n', m=n,p=q, \ell =k} \\
      & =     - 2i \sum_{(k',n', k, n, q)} k'k((q+n) -k\tan^2 \psi) \lambda [a\mathbf{b}^+_{k'n'n'} \mathbf{a}^-_{qqk} \mathbf{b}^+_{knn}  \\
                                              & - a\mathbf{b}^+_{k'n'n'} \mathbf{a}^+_{qqk} \mathbf{b}^-_{knn} + a\mathbf{b}^-_{k'n'n'} \mathbf{a}^+_{qqk} \mathbf{b}^-_{knn} - a\mathbf{b}^-_{k'n'n'} \mathbf{a}^-_{qqk} \mathbf{b}^+_{knn}] \\
                                               
  \end{array}
 \end{equation*}  

\begin{equation*}
   \begin{array}{ll}
      & A_2  [ \slashed{\partial}_3, A_1 ] - A_1 [\slashed{\partial}_3, A_2]|_{m'=n', m=n,p=q, \ell =k} \\
      & =     - 4i \sum_{(k',n', k, n, q)} (2n'n \cot^2 \psi - n'k -k'n) \lambda [a\mathbf{b}^+_{k'n'n'} \mathbf{a}^-_{qqk} \mathbf{b}^+_{knn}  \\
                                              & + a\mathbf{b}^+_{k'n'n'} \mathbf{a}^+_{qqk} \mathbf{b}^-_{knn} + a\mathbf{b}^-_{k'n'n'} \mathbf{a}^+_{qqk} \mathbf{b}^-_{knn} + a\mathbf{b}^-_{k'n'n'} \mathbf{a}^-_{qqk} \mathbf{b}^+_{knn}] \\
                                              & + k'k\tan^2 \psi [ a\mathbf{b}^+_{k'n'n'} \mathbf{a}^+_{qqk} \mathbf{b}^-_{knn} +  a\mathbf{b}^-_{k'n'n'} \mathbf{a}^-_{qqk} \mathbf{b}^+_{knn} ]
                                               
  \end{array}
 \end{equation*}

\begin{equation*}
   \begin{array}{ll}
      & A_3  [ \slashed{\partial}_1, A_2 ] - A_3 [\slashed{\partial}_2, A_1]|_{m'=n', m=n,p=q, \ell =k} \\
      & =     - 2i \sum_{(k',n', k, n, q)} (k'k(q+n+1)) [a\mathbf{b}^+_{k'n'n'} \mathbf{a}^+_{qqk} \mathbf{b}^-_{knn}  \\
                                              & - a\mathbf{b}^+_{k'n'n'} \mathbf{a}^-_{qqk} \mathbf{b}^+_{knn} - a\mathbf{b}^-_{k'n'n'} \mathbf{a}^+_{qqk} \mathbf{b}^-_{knn} + a\mathbf{b}^-_{k'n'n'} \mathbf{a}^-_{qqk} \mathbf{b}^+_{knn}] \\
                                               
  \end{array}
 \end{equation*}  
So we have  
\begin{equation*} 
  \begin{array}{ll}
     & \varepsilon^{ijk} A_i [ \slashed{\partial}_j, A_k]|_{m'=n', m=n,p=q, \ell =k}  \\   
     & =     - 2i \sum_{(k',n', k, n, q)} [(4n'n \cot^2 \psi - 2n'k -2k'n + \\
     & k'k((q+n) -k\tan^2 \psi))  \lambda  -  k'k(q+n+1)] a\mathbf{b}^+_{k'n'n'} \mathbf{a}^-_{qqk} \mathbf{b}^+_{knn}  \\
     & +[(4n'n \cot^2 \psi - 2n'k -2k'n  - k'k((q+n) -k\tan^2 \psi))  \lambda \\
     & + 2k'k\tan^2 \psi     +  k'k(q+n+1)] a\mathbf{b}^+_{k'n'n'} \mathbf{a}^+_{qqk} \mathbf{b}^-_{knn}  \\
      & +[(4n'n \cot^2 \psi - 2n'k -2k'n  + k'k((q+n) -k\tan^2 \psi))  \lambda \\
     & -  k'k(q+n+1)] a\mathbf{b}^-_{k'n'n'} \mathbf{a}^+_{qqk} \mathbf{b}^-_{knn}  \\
     & +[(4n'n \cot^2 \psi - 2n'k -2k'n  - k'k((q+n) -k\tan^2 \psi))  \lambda \\
     & + 2k'k\tan^2 \psi     +  k'k(q+n+1)] a\mathbf{b}^-_{k'n'n'} \mathbf{a}^-_{qqk} \mathbf{b}^+_{knn}  
                                             
  \end{array}
\end{equation*}  

Now we compute the 3-forms $A_iA_jA_k$ and combine them together, we further simplify them by restricting  to the diagonal region,

\begin{equation*}
   \begin{array}{ll}
   &  A_1A_2A_3 - A_2A_1A_3|_{m'=n', \tilde{p}= \tilde{q}, \tilde{\ell} = \tilde{k}, m= n} \\
   & = -2i \sum [ (k'\tilde{n} - n'\tilde{k}) \lambda a \mathbf{b}^+_{k'n'n'} \mathbf{a}^-_{\tilde{q}\tilde{q}\tilde{k}} \mathbf{b}^+_{\tilde{k}\tilde{n}\tilde{n}} \\
   & + (k'\tilde{n} - k'\tilde{k} \tan^2 \psi + m' \tilde{k}) \lambda  a \mathbf{b}^+_{k'n'n'} \mathbf{a}^+_{\tilde{q}\tilde{q}\tilde{k}} \mathbf{b}^-_{\tilde{k}\tilde{n}\tilde{n}} \\
   & - (k'\tilde{m} - k'\tilde{k} \tan^2 \psi + n' \tilde{k}) \lambda  a \mathbf{b}^-_{k'n'n'} \mathbf{a}^-_{\tilde{q}\tilde{q}\tilde{k}} \mathbf{b}^+_{\tilde{k}\tilde{n}\tilde{n}} \\
   & - (k'\tilde{m} - m'\tilde{k}) \lambda a \mathbf{b}^-_{k'n'n'} \mathbf{a}^+_{\tilde{q}\tilde{q}\tilde{k}} \mathbf{b}^-_{\tilde{k}\tilde{n}\tilde{n}}] [(\pm k) a \mathbf{b}^\pm_{knn}]  \\

   \end{array}
\end{equation*}

\begin{equation*}
   \begin{array}{ll}
      & A_2A_3A_1 - A_1A_3A_2|_{m'=n', \tilde{m} = \tilde{n}, \tilde{p} = \tilde{q}, \tilde{\ell} = \tilde{k}, \ell = k}  \\
      & = -2i \sum (\pm \tilde{k}) [(n'k - k'n) \lambda  a \mathbf{b}^+_{k'n'n'} \mathbf{a}^\mp_{\tilde{q}\tilde{q}\tilde{k}} \mathbf{b}^\pm_{\tilde{k}\tilde{n}\tilde{n}} \mathbf{a}^-_{qqk} \mathbf{b}^+_{knn}  \\
      & + (k'k \tan^2 \psi - k'n - m'k) \lambda  a \mathbf{b}^+_{k'n'n'} \mathbf{a}^\mp_{\tilde{q}\tilde{q}\tilde{k}} \mathbf{b}^\pm_{\tilde{k}\tilde{n}\tilde{n}} \mathbf{a}^+_{qqk} \mathbf{b}^-_{knn} \\
      & - (k'k \tan^2 \psi - k'm - n'k) \lambda  a \mathbf{b}^-_{k'n'n'} \mathbf{a}^\mp_{\tilde{q}\tilde{q}\tilde{k}} \mathbf{b}^\pm_{\tilde{k}\tilde{n}\tilde{n}} \mathbf{a}^-_{qqk} \mathbf{b}^+_{knn} \\
      & - (m'k - k'm) \lambda  a \mathbf{b}^-_{k'n'n'} \mathbf{a}^\mp_{\tilde{q}\tilde{q}\tilde{k}} \mathbf{b}^\pm_{\tilde{k}\tilde{n}\tilde{n}} \mathbf{a}^+_{qqk} \mathbf{b}^-_{knn}]  \\

   \end{array}
\end{equation*}

\begin{equation*}
   \begin{array}{ll}
      & A_3A_1A_2 - A_3A_2A_1|_{m'=n', \tilde{m} = \tilde{n}, {p} = {q}, \ell = k} \\
      & = -2i \sum [(\pm k') a \mathbf{b}^\pm_{k'n'n'} ]  [ (\tilde{k}{n} - \tilde{n}{k}) \lambda a \mathbf{b}^+_{\tilde{k}\tilde{n}\tilde{n}} \mathbf{a}^-_{{q}{q}{k}} \mathbf{b}^+_{{k}{n}{n}} \\
   & + (\tilde{k}{n} - \tilde{k}{k} \tan^2 \psi + \tilde{m}{k}) \lambda  a \mathbf{b}^+_{\tilde{k}\tilde{n}\tilde{n}} \mathbf{a}^+_{{q}{q}{k}} \mathbf{b}^-_{{k}{n}{n}} \\
   & - (\tilde{k}{m} - \tilde{k}{k} \tan^2 \psi + \tilde{n}{k}) \lambda  a \mathbf{b}^-_{\tilde{k}\tilde{n}\tilde{n}} \mathbf{a}^-_{{q}{q}{k}} \mathbf{b}^+_{{k}{n}{n}} \\
   & - (\tilde{k}{m} - \tilde{m}{k}) \lambda a \mathbf{b}^-_{\tilde{k}\tilde{n}\tilde{n}} \mathbf{a}^+_{{q}{q}{k}} \mathbf{b}^-_{{k}{n}{n}}]   \\

   \end{array}
\end{equation*}

A direct computation shows that the alternating sum of the 3-forms $A_iA_jA_k$ is trivial,  

\begin{equation*}
  \varepsilon^{ijk} A_iA_jA_k|_{m'=n', \tilde{p} = \tilde{q}, \tilde{m} = \tilde{n}, \tilde{\ell} = \tilde{k}, p=q, \ell = k} =0
\end{equation*}
 
Now the residual trace over the diagonal region is 
\begin{equation*}
  \begin{array}{rl}
    &  Res_{z=0}Tr(\varepsilon^{ijk}(3A_i [ \slashed{\partial}_j, A_k]|_{m'=n', m=n,p=q, \ell =k})|\mathcal{D}_1|^{-3-z}) \\
    & = -6i  Res_{z=0}   \sum_{s}  
     <\tilde{\Phi}^{\pm s},  \varepsilon^{ijk}(3A_i [ \slashed{\partial}_j, A_k]|_R) (s+\frac{3}{2})^{-3-z}\tilde{\Phi}^{\pm s}  > \\
    & = -12i \{ \sum  [( - 2n'k -2k'n +  k'k(q+n))  \lambda -  k'k(q+n+1)] \\
    &  [{a'}_{q'q'k'}b_{k'n'n'} {a'}_{qqk} {b}_{knn} + a_{q'q'k'}{b'}_{k'n'n'} {a}_{qqk} {b'}_{knn}]+  \\
     & [(- 2n'k -2k'n  - k'k(q+n)   )  \lambda  +  k'k(q+n+1)]   \\
     & [{a}_{q'q'k'}b'_{k'n'n'} {a'}_{qqk} {b}_{knn} + a'_{q'q'k'}{b}_{k'n'n'} {a}_{qqk} {b'}_{knn}] \} \\
    & Res_{z=0} [\zeta_H(z+1, 3/2) -\frac{1}{4} \zeta_H(z+3, 3/2) ] \\
   &  = - 12i \sum [( - 2n'k -2k'n +  k'k(q+n) )  \lambda -  k'k(q+n+1)]  \\
    & [{a'}_{q'q'k'}b_{k'n'n'} {a'}_{qqk} {b}_{knn} + a_{q'q'k'}{b'}_{k'n'n'} {a}_{qqk} {b'}_{knn}]+  \\
     & [( - 2n'k -2k'n  - k'k(q+n))  \lambda +  k'k(q+n+1)]   \\
     & [{a}_{q'q'k'}b'_{k'n'n'} {a'}_{qqk} {b}_{knn} + a'_{q'q'k'}{b}_{k'n'n'} {a}_{qqk} {b'}_{knn}] 
  \end{array}
\end{equation*}
where this formal series is over all positive integers $(k',n',q', k, n, q)$.
In the middle, $\cot^2\psi$ and $\tan^2 \psi$ shift the index of the orthogonal basis, so they can be canceled out by taking the trace.
Finally we use the fact $\zeta_H(s, 3/2)$ has  residue 1 at its simple pole $s= 1$. 

\end{proof}

\begin{rmk}
    Here we did not use the self-adjoint condition of the connection 1-form, i.e., $A = A^*$, since we do not see it will simplify the computation
     a lot.
\end{rmk}

\begin{rmk}
   Similar to the method used in \cite{P1201}, another way to compute the Chern--Simons action is to
   take advantage of the representation of  $ SU_\lambda(2)= {S}^3_\theta$ and construct a spectral triple, which is essentially the same as the generalized Dirac geometry.
   In this paper, we compute directly on the elements in  $C^\infty({S}^3_\theta)$,
   which is convenient for the comparison with different Dirac operators.  
\end{rmk}

\section{Choice of Dirac operator} \label{Diracop}

In this section we give another two spectral triples on $S^3_\theta$ with the same Dirac Laplacian spectrum as in the classical 3-sphere.
The Chern--Simons action will be computed and we conclude that it depends on the choice of Dirac operators.

If we consider  the round metric  on $S^3_\theta$ in Hopf coordinates, 
\begin{equation}
    G =   d\psi^2 + \cos^2\psi\,du du^* + \sin^2\psi\,dv dv^* 
\end{equation}
we get another Dirac operator by direct computation,
\begin{equation}
  \mathcal{D}_2  = \sec \psi \, \delta_1 \sigma_1 +  \csc \psi \, \delta_2 \sigma_2 + i  
  [{\partial_\psi} + \frac{1}{2} (\cot \psi - \tan \psi)] \sigma_3
\end{equation}
$\mathcal{D}_2$ can also be obtained by restricting the Dirac operator over $S^4_\theta$ \cite{CL01} onto the equator $S^3_\theta$ 
when we fix the second angle to be a constant.
Notice that the classical Laplace--Beltrami operator corresponds to 
\begin{equation}
 {\mathcal{D}_2^2}'= \sec^2 \psi \, \delta_1^2 + \csc^2 \psi \, \delta_2^2 - \partial^2_\psi - 2 \cot (2\psi) \,\partial_\psi 
\end{equation}
again $ {\mathcal{D}_2^2}' $ is  obtained by dropping the constant term in $\mathcal{D}_2^2$,
and its eigenvalues are given by
\begin{equation}
  {\mathcal{D}_2^2}' \tilde{\phi}^m_{l,j} = (m^2+2m) \tilde{\phi}^m_{l,j} 
\end{equation}
with multiplicity $(m + 1)^2$.

Now we have a second spectral triple  $(C^\infty(S^3_\theta), L^2(S^3_\theta),  \mathcal{D}_2)$ on the quantum 3-sphere, 
and  one could double it and consider the augmented spectral triple $(C^\infty(S^3_\theta) \otimes M_2(\mathbb{C}), L^2(S^3_\theta) \otimes \mathbb{C}^2,  \mathcal{D}_2 \otimes I_2)$.
However, in order to compare with $(C^\infty(S^3_\theta), L^2(S^3_\theta, \mathbf{S}), \mathcal{D}_1 )$ on the same footing,
we  consider the spectral triple only with  the Hilbert space 
  augmented, i.e.  $(C^\infty(S^3_\theta), L^2(S^3_\theta) \otimes \mathbb{C}^2,  \mathcal{D}_2)$. Further assume that 
  $L^2(S^3_\theta) \otimes \mathbb{C}^2$ is equipped with a Hilbert basis $\tilde{\phi}^m_{lj} \otimes e_i$ 
  where $\{ e_i\} \,(i =1,2)$ is the standard basis in $\mathbb{C}^2$.

The commutators of the Dirac operator $\mathcal{D}_2$ with the generators are 
\begin{equation*}
           [\mathcal{D}_2, \alpha] =  u \sigma_1  - i u \sin \psi   \sigma_3  = \begin{pmatrix}
                                                                                                  - i u \sin \psi &  u\\
                                                                                                  u &  i u \sin \psi
                                                                                                \end{pmatrix} 
\end{equation*}
\begin{equation*}
             [\mathcal{D}_2 , \beta ] =   v \sigma_2  + i v \cos \psi \sigma_3   = \begin{pmatrix}
                                                                                                   i v \cos \psi &  -i v\\
                                                                                                  i v & - i v \cos \psi
                                                                                                \end{pmatrix} 
\end{equation*}
\begin{equation*}
            [ \mathcal{D}_2 , \alpha^* ] =  - u^*  \sigma_1   - i u^* \sin\psi \sigma_3  = \begin{pmatrix}
                                                                                                  - i u^* \sin \psi &  -u^*\\
                                                                                                  -u^* &  i u^* \sin \psi
                                                                                                \end{pmatrix}  
\end{equation*}
\begin{equation*}
             [ \mathcal{D}_2, \beta ^*] = -  v^* \sigma_2 + i v^* \cos \psi \sigma_3  = \begin{pmatrix}
                                                                                                   i v^* \cos \psi &  i v^*\\
                                                                                                  -i v^* & - i v^* \cos \psi
                                                                                                \end{pmatrix} 
\end{equation*}
 $[\mathcal{D}_2, a]$ is a bounded operator for any $a \in C^\infty(S^3_\theta)$ and the spectral triple
$(C^\infty(S^3_\theta), L^2(S^3_\theta) \otimes \mathbb{C}^2, \mathcal{D}_2)$ is also  3-summable regular 
since it gives an isospectral deformation.

\begin{lemma}
  The spectral triple
$(C^\infty(S^3_\theta), L^2(S^3_\theta) \otimes \mathbb{C}^2, \mathcal{D}_2)$
 has simple dimension spectrum $\{ 3 \}$.
\end{lemma}

\begin{proof}
  
Let us look at the spectral zeta function, 
\begin{equation*} 
   \begin{array}{ll}
      & Tr(b |{\mathcal{D}_2}|^{-z})  \\
      & = \sum_{m, k, \ell} <\tilde{\phi}^m_{k,\ell}, b ({\mathcal{D}_2^2})^{-z/2} \tilde{\phi}^m_{k,\ell} > \\
                 & = \sum_{m \geq 0 } (m+1)^2 (m^2 +2m - \cot^2 2\psi)^{-z/2} <\tilde{\phi}^m_{k,\ell}, b  \tilde{\phi}^m_{k,\ell} > \\
               & =b_0 \sum_{m \geq 0} (m+1)^2 (m^2 +2m -\cot^2 2\psi)^{-z/2}   \\
               & = b_0 \sum_{m \geq 0} (m+1)^2 [(m+1)^2 -\csc^2 2\psi]^{-z/2}   \\
               & = b_0 \sum_{n \geq 1 } n^2 (n^2 -\csc^2 2\psi)^{-z/2}   \\
   \end{array}
\end{equation*}  
For fixed $\psi$, there exists a smallest $n_0(\psi)$ such that $n_0^2 > \csc^2 2\psi $. 
On the other hand, we know the binomial expansion for $|w| <1$,
\begin{equation*}
  (1-w)^{-s} = \sum_{k =0}^\infty \frac{\Gamma(s+k)}{\Gamma(s) k!} w^k
\end{equation*}
We could modify the first $n_0$ terms since they don't change the singular points and residues of the spectral zeta function, 
and we write it in terms of  Riemann zeta function,
\begin{equation*} 
   \begin{array}{ll}
       & Tr(b |{\mathcal{D}_2}|^{-z}) \\
       & = b_0 \left( \sum_{ 1 \leq n \leq n_0 } + \sum_{n \geq n_0 } \right) n^2 (n^2 -\csc^2 2\psi)^{-z/2} \\
          & \sim  b_0 \sum_{n \geq n_0}   n^2 (n^2 -\csc^2 2\psi)^{-z/2}   \\
          & =  b_0 \sum_{n \geq n_0}   n^{2-z} (1 -\csc^2 2\psi/ n^2)^{-z/2}   \\
          & =  b_0 \sum_{n \geq n_0}   n^{2-z} \sum_{k \geq 0} \frac{\Gamma(k+ z/2)}{\Gamma(z/2) k!} \left( \frac{\csc^2 2\psi}{ n^2} \right)^k   \\
      & =  b_0 \sum_{k \geq 0} {\csc^{2k} 2\psi} \frac{\Gamma(k+ z/2)}{\Gamma(z/2) k!}   \sum_{n \geq n_0}   n^{2-z-2k}    \\
       & \sim  b_0 \sum_{k \geq 0} {\csc^{2k} 2\psi} \frac{\Gamma(k+ z/2)}{\Gamma(z/2) k!}   \zeta_R(z+2k-2)    \\
   \end{array}
\end{equation*}  
So the spectral zeta function only has a simple pole at $z= 3$. 
\end{proof}

\begin{thm}
   The Chern-Simons action on $S^3_\theta$ with respect to the spectral triple $(C^\infty(S^3_\theta),  L^2(S^3_\theta) \otimes \mathbb{C}^2, \mathcal{D}_2 )$ is trivial,  
\begin{equation}
    S_{CS}(A)   =    0
\end{equation}
\end{thm}

\begin{proof}
 
Let $\partial_1 = \sec \psi \delta_1$, $\partial_2 =  \csc \psi \delta_2$ and $\partial_3 = i(\partial_\psi + \cot 2\psi)$, 
and a connection $A = a [\mathcal{D}_2, b]$ with  $a, b$ as before. 
\begin{equation*}
  \begin{array}{ll}
     a   = \sum_{(p, q, \ell) } \mathbf{a}^+_{pq\ell} +  \mathbf{a}^-_{pq\ell} =\sum_{(p, q, \ell) }  a_{pq\ell} \beta^p {\beta^*}^q \alpha^\ell + a'_{pq\ell} \beta^p {\beta^*}^q {\alpha^*}^\ell, \\
     b   = \sum_{(k,m,n) } \mathbf{b}^+_{kmn} + \mathbf{b}^-_{kmn} =\sum_{(k,m,n) }b_{kmn}\alpha^k \beta^m {\beta^*}^n + b'_{kmn}{\alpha^*}^k \beta^m {\beta^*}^n
  \end{array}
\end{equation*} 
The components of the connection are 
\begin{equation*}
  \begin{array}{ll}
   A_1  = \sum_{(k,m,n)}k  \sec \psi a(\mathbf{b}^+_{kmn} - \mathbf{b}^-_{kmn}) \\
    A_2 =  \sum_{(k,m,n)}(m-n)  \csc \psi a(\mathbf{b}^+_{kmn}  + \mathbf{b}^-_{kmn}) \\
    A_3  = i \sum_{(k,m,n)} [(m+n) \cot \psi - k \tan \psi ]  a(\mathbf{b}^+_{kmn}  + \mathbf{b}^-_{kmn}) 
  \end{array}
\end{equation*}  
and direct computation gives
\begin{equation*}
  \begin{array}{rl}
  A_1 [\partial_2, A_3]  = & i \sum_{} k'(p-q + m-n) [(m+n) \cot \psi - k \tan \psi] \\
                           & \sec \psi \csc \psi\, a(\mathbf{b}^+_{k'm'n'} - \mathbf{b}^-_{k'm'n'})\,  {a}(\mathbf{b}^+_{kmn}  + \mathbf{b}^-_{kmn})                 
     \end{array}
\end{equation*}
\begin{equation*}
  \begin{array}{rl}
  A_1 [\partial_3,A_2]  = & i\sum_{ } k'(m-n) [(p+q+ m+n-1) \cot \psi - ( \ell + k) \tan \psi] \\
                          &  \sec \psi \csc \psi \,  a(\mathbf{b}^+_{k'm'n'} - \mathbf{b}^-_{k'm'n'})  {a}(\mathbf{b}^+_{kmn}  + \mathbf{b}^-_{kmn})                  
    \end{array}
\end{equation*}
\begin{equation*}
  \begin{array}{rl}
   A_2[\partial_1,A_3]  = & i\sum_{ }  (m' - n')(\pm \ell  \pm k)[(m+n) \cot \psi - k \tan \psi]\\
                          &  \sec \psi \csc \psi \,a(\mathbf{b}^+_{k'm'n'} + \mathbf{b}^-_{k'm'n'})  \mathbf{a}^\pm_{pq\ell} \mathbf{b}^\pm_{kmn}                   
   \end{array}
\end{equation*}
\begin{equation*}
  \begin{array}{rl}
   A_2[\partial_3,A_1]  = & i\sum_{ } k(m'-n') [(p+q +m +n) \cot \psi - (\ell +k -1) \tan \psi] \\
                          &  \sec \psi \csc \psi  \, a(\mathbf{b}^+_{k'm'n'} + \mathbf{b}^-_{k'm'n'})  {a}(\mathbf{b}^+_{kmn}  - \mathbf{b}^-_{kmn})                  
  \end{array} 
\end{equation*}
\begin{equation*}
  \begin{array}{rl}
   A_3[\partial_1,A_2] =  & i\sum_{ }  (m-n)  (\pm\ell \pm k)[(m'+n') \cot \psi - k' \tan \psi ) \\
                            &  \sec \psi \csc \psi  \, a(\mathbf{b}^+_{k'm'n'} + \mathbf{b}^-_{k'm'n'}) \mathbf{a}^\pm_{pq\ell} \mathbf{b}^\pm_{kmn}   
  \end{array}
\end{equation*}
\begin{equation*}
  \begin{array}{rl}
   A_3[\partial_2,A_1]  =   & i\sum_{ }   k(p-q+m -n)[(m'-n') \cot \psi - k' \tan \psi ] \\
                            & \sec \psi \csc \psi   \,a(\mathbf{b}^+_{k'm'n'} + \mathbf{b}^-_{k'm'n'})  {a}(\mathbf{b}^+_{kmn}  - \mathbf{b}^-_{kmn})               
  \end{array}
\end{equation*}
If we only consider the diagonal region where $m =n$, $p=q$, $m'=n'$ and $\ell = k$, then it is clear that 
\begin{equation*} 
  \begin{array}{ll}
     \varepsilon^{ijk} A_i[\partial_j,A_k]|_{m'=n', m=n, p=q, \ell =k} = 0 
   \end{array}
\end{equation*}
By direct computation, we have  
\begin{equation*}
   \begin{array}{ll}
      & \varepsilon_{ijk} A_i A_j A_k  \\
      & = i\sum k' [(2n\tilde{m} - 2 \tilde{n}m) \cot \psi - (\tilde{m} - \tilde{n}) k \tan \psi + (m-n) \tilde{k} \tan \psi] \sec \psi \csc \psi \\
      & (\mathbf{a}^-_{p'q'\ell'} \mathbf{b}^+_{k'm'n'} - \mathbf{a}^+_{p'q'\ell'} \mathbf{b}^-_{k'm'n'})(\mathbf{a}^+_{\tilde{p}\tilde{q}\tilde{\ell}} \mathbf{b}^-_{\tilde{k}\tilde{m}\tilde{n}} + \mathbf{a}^-_{\tilde{p}\tilde{q}\tilde{\ell}} \mathbf{b}^+_{\tilde{k}\tilde{m}\tilde{n}} )(\mathbf{a}^+_{pq\ell} \mathbf{b}^-_{kmn} + \mathbf{a}^-_{pq\ell} \mathbf{b}^+_{kmn} )   \\   
       & - \tilde{k} [(2n{m}' - 2 {n}'m) \cot \psi - ({m}' - {n}') k \tan \psi + (m-n) {k}' \tan \psi] \sec \psi \csc \psi \\
      & (\mathbf{a}^+_{p'q'\ell'} \mathbf{b}^-_{k'm'n'} + \mathbf{a}^-_{p'q'\ell'} \mathbf{b}^+_{k'm'n'})(\mathbf{a}^-_{\tilde{p}\tilde{q}\tilde{\ell}} \mathbf{b}^+_{\tilde{k}\tilde{m}\tilde{n}} - \mathbf{a}^+_{\tilde{p}\tilde{q}\tilde{\ell}} \mathbf{b}^-_{\tilde{k}\tilde{m}\tilde{n}} )(\mathbf{a}^+_{pq\ell} \mathbf{b}^-_{kmn} + \mathbf{a}^-_{pq\ell} \mathbf{b}^+_{kmn} )   \\ 
     & + {k} [(2\tilde{n}{m}' - 2 {n}'\tilde{m}) \cot \psi - ({m}' - {n}') \tilde{k} \tan \psi + (\tilde{m}-\tilde{n}) {k}' \tan \psi] \sec \psi \csc \psi \\
      & (\mathbf{a}^+_{p'q'\ell'} \mathbf{b}^-_{k'm'n'} + \mathbf{a}^-_{p'q'\ell'} \mathbf{b}^+_{k'm'n'})(\mathbf{a}^+_{\tilde{p}\tilde{q}\tilde{\ell}} \mathbf{b}^-_{\tilde{k}\tilde{m}\tilde{n}} + \mathbf{a}^-_{\tilde{p}\tilde{q}\tilde{\ell}} \mathbf{b}^+_{\tilde{k}\tilde{m}\tilde{n}} )(\mathbf{a}^-_{pq\ell} \mathbf{b}^+_{kmn} - \mathbf{a}^+_{pq\ell} \mathbf{b}^-_{kmn} )   \\    
   \end{array}
\end{equation*}
When restricted to the diagonal region where $m=n$, $m'=n'$ and $\tilde{m} = \tilde{n}$, we get
$$
 \varepsilon_{ijk} A_i A_j A_k|_{m'=n', \tilde{m} = \tilde{n}, m= n} = 0 
$$

\end{proof}

From the first spectral triple,  we have seen the orthogonal framing of $T_eS^3$ in Hopf coordinates,
\begin{equation}
  \{ \partial_{\xi_1} + \partial_{\xi_2}, \partial_\eta,  \tan \eta \partial_{\xi_1}  - \cot \eta \partial_{\xi_2}  \}  
\end{equation}
where the first vector field is tangent to the Hopf fiber as mentioned before.
It is possible to define a third Dirac operator on $S^3_\theta$ by
\begin{equation}
   \mathcal{D}_3 = i \partial_\psi \sigma_1 - ( \tan \psi \delta_1  - \cot \psi \delta_2 )   \sigma_2 - (\delta_1 + \delta_2) \sigma_3
\end{equation}
and its Dirac Laplacian corresponds to the round metric as well,
\begin{equation}
   \mathcal{D}^2_3 = - \partial_\psi^2 + \sec^2 \psi  \delta_1^2 + \csc^2 \psi  \delta_2^2
\end{equation}
Thus a third spectral triple can be defined as $(C^\infty(S^3_\theta),  L^2(S^3_\theta) \otimes \mathbb{C}^2, \mathcal{D}_3 )$, 
which is a 3-summable regular spectral triple with simple dimension spectrum $\{ 3 \}$ as in the second spectral triple.

\begin{thm}
   The Chern--Simons action on $S^3_\theta$ with respect to the spectral triple $(C^\infty(S^3_\theta),  L^2(S^3_\theta) \otimes \mathbb{C}^2, \mathcal{D}_3 )$
    is given by   
\begin{equation} 
   \begin{array}{ll}
    & S_{CS}(A)
     =    \sum k'k [ {a}_{q'q'k'} {b}'_{k'n'n'}{a}_{qqk}{b}'_{knn} -  {a}_{q'q'k'} {b}'_{k'n'n'}{a}'_{qqk}{b}_{knn}\\
    &  +  {a}'_{q'q'k'} {b}_{k'n'n'}{a}'_{qqk}{b}_{knn}  -  {a}'_{q'q'k'} {b}_{k'n'n'}{a}_{qqk}{b}'_{knn} ] \\ 
   \end{array}
\end{equation}
\end{thm}

\begin{rmk}
   In particular, if the coefficients of $a$ and $b$ are symmetric, then the above Chern--Simons action is trivial.
\end{rmk}

\begin{proof}
 
Let $\hat{\partial}_1 = i \partial_\psi $, $\hat{\partial}_2 =  - ( \tan \psi \delta_1  - \cot \psi \delta_2 )$ and 
$\hat{\partial}_3 = - (\delta_1 + \delta_2)$, 
and a connection $A = a [\mathcal{D}_3, b]$ with 
\begin{equation*}
  \begin{array}{ll}
     a   = \sum_{(p, q, \ell) } \mathbf{a}^+_{pq\ell} +  \mathbf{a}^-_{pq\ell} =\sum_{(p, q, \ell) }  a_{pq\ell} \beta^p {\beta^*}^q \alpha^\ell + a'_{pq\ell} \beta^p {\beta^*}^q {\alpha^*}^\ell, \\
     b   = \sum_{(k,m,n) } \mathbf{b}^+_{kmn} + \mathbf{b}^-_{kmn} =\sum_{(k,m,n) }b_{kmn}\alpha^k \beta^m {\beta^*}^n + b'_{kmn}{\alpha^*}^k \beta^m {\beta^*}^n
  \end{array}
\end{equation*} 

The components of the connection are given by
\begin{equation*}
   \begin{array}{ll}
      A_1 = i \sum_{ (k,m,n)} [(m+n) \cot \psi - k \tan \psi] a (\mathbf{b}^+_{kmn} + \mathbf{b}^-_{kmn}  ) \\
      A_2 =  \sum_{ (k,m,n)} [ (m-n) \cot \psi - (\pm k) \tan \psi] a \mathbf{b}^{\pm}_{kmn} \\
       A_3= -\sum_{(k,m,n)} (\pm k + m- n ) a \mathbf{b}^{\pm}_{kmn}  
   \end{array}
\end{equation*}

We compute the terms $A_i[\hat{\partial}_j, A_j]$ as before,
\begin{equation*}
   \begin{array}{ll}
       & A_1 [ \hat{\partial}_2, A_3]  \\ 
       & = i\sum_{ } [(m'+n') \cot \psi - k' \tan \psi] [\pm k + (m-n)] [ (\pm \ell \pm k) \tan \psi   \\
       & -  (p-q + m-n) \cot \psi ] a(\mathbf{b}^+_{k'm'n'} +\mathbf{b}^-_{k'm'n'} )\mathbf{a}^\pm_{pq\ell}\mathbf{b}^\pm_{kmn}   
  \end{array}
 \end{equation*} 
\begin{equation*}
   \begin{array}{ll}
       & A_1 [  \hat{\partial}_3, A_2] \\ 
       & = i \sum_{ } [(m'+n') \cot \psi - k' \tan \psi)] [\pm k \tan \psi - (m-n)\cot \psi] \\
       & [(\pm\ell \pm k) + (p-q +m -n)] a(\mathbf{b}^+_{k'm'n'} + \mathbf{b}^-_{k'm'n'}) \mathbf{a}^\pm_{pq\ell}\mathbf{b}^\pm_{kmn}  
  \end{array}
 \end{equation*}  
\begin{equation*}
   \begin{array}{ll}
         & A_2  [ \hat{\partial}_1, A_3] \\
         & = i \sum_{ }   [ \pm k' \tan \psi - (m' - n') \cot \psi ] [\pm k +(m - n)] \\
         &  [  (p+q +m +n) \cot \psi - ( \ell + k) \tan \psi ]   a\mathbf{b}^\pm_{k'm'n'} {a} \mathbf{b}^\pm_{kmn} 
  \end{array}
 \end{equation*}
\begin{equation*}
   \begin{array}{ll}
       & A_2  [ \hat{\partial}_3, A_1]  \\
       & =  i \sum_{ }  ( \pm k' \tan \psi - (m'-n') \cot \psi  ] [(m+n) \cot \psi - k \tan \psi ] \\
        &   [(\pm \ell \pm k) + (p -q + m -n)]  a\mathbf{b}^\pm_{k'm'n'}\mathbf{a}^\pm_{pq\ell}\mathbf{b}^\pm_{kmn} 
  \end{array}
 \end{equation*}
\begin{equation*}
   \begin{array}{ll}
         & A_3  [  \hat{\partial}_1, A_2] \\
         & = i \sum_{ } (\pm k' + m' - n')  \{ \pm k [ (p + q + m + n + 1)   -( \ell + k - 1) \tan^2 \psi] \\
          & - (m-n) [(p + q + m + n -1) \cot^2 \psi - (\ell + k +1)]\}  a\mathbf{b}^\pm_{k'm'n'} {a} \mathbf{b}^\pm_{kmn} 
        \end{array}
 \end{equation*}
\begin{equation*}
   \begin{array}{ll}
      &  A_3  [ \hat{\partial}_2, A_1]\\
      & = i \sum_{ } (\pm k'+ m'- n')  [(m +n) \cot \psi - k \tan \psi ] \\
      &    [  (\pm \ell \pm k) \tan \psi - (p-q + m -n) \cot \psi ]  a\mathbf{b}^\pm_{k'm'n'}\mathbf{a}^\pm_{pq\ell}\mathbf{b}^\pm_{kmn}  
  \end{array}
 \end{equation*}  
 
Again we consider the diagonal region where $m'= n'$, $m=n$, $p= q$ and $\ell = k$, we obtain the sum 
\begin{equation*} 
  \begin{array}{ll}
    & \varepsilon^{ijk} A_i [ \hat{\partial}_j, A_k]|_{m'=n', m=n, p=q, \ell =k} \\
    & =  i \sum_{(p',q', \ell',k',n,q,k )}  k'k \sec^2 \psi [ \mathbf{a}^+_{p'q'\ell'} \mathbf{b}^-_{k'n'n'}\mathbf{a}^+_{qqk}\mathbf{b}^-_{knn} \\
    &-  \mathbf{a}^+_{p'q'\ell'} \mathbf{b}^-_{k'n'n'}\mathbf{a}^-_{qqk}\mathbf{b}^+_{knn}  +  \mathbf{a}^-_{p'q'\ell'} \mathbf{b}^+_{k'n'n'}\mathbf{a}^-_{qqk}\mathbf{b}^+_{knn} \\
    & -  \mathbf{a}^-_{p'q'\ell'} \mathbf{b}^+_{k'n'n'}\mathbf{a}^+_{qqk}\mathbf{b}^-_{knn} ] 
     
  \end{array}
\end{equation*} 

Next, we compute the 3-forms $A_iA_jA_k$ and combine them together, 
\begin{equation*}
   \begin{array}{ll}
      & A_1A_2A_3 - A_1A_3A_2  \\
      & =  i \sum [ (m'+n')(\pm \tilde{k})(m-n) \csc^2 \psi  -k'(\pm \tilde{k})(m-n) \sec^2 \psi \\
      & - (m'+n')(\tilde{m}- \tilde{n})(\pm k) \csc^2 \psi 
       + k'(\tilde{m} - \tilde{n})(\pm k) \sec^2 \psi ] aba\mathbf{b}^\pm_{\tilde{k}\tilde{n}\tilde{n}}a\mathbf{b}^\pm_{knn}
   \end{array}
\end{equation*}

\begin{equation*}
   \begin{array}{ll}
      & A_3A_1A_2 - A_2A_1A_3 \\
      & = i \sum [(m'-n')(\tilde{m}+ \tilde{n})(\pm k) \csc^2 \psi  - (m'-n')\tilde{k} (\pm k) \sec^2 \psi \\
      & - (\pm k')(\tilde{m} + \tilde{n})(m-n) \csc^2 \psi 
      + (\pm k')\tilde{k} (m-n) \sec^2 \psi ]  a\mathbf{b}^\pm_{{k}'{n}'{n}'}aba\mathbf{b}^\pm_{knn}
      
   \end{array}
\end{equation*}

\begin{equation*}
   \begin{array}{ll}
      & A_2A_3A_1 - A_3A_2A_1 \\
      & = i \sum [(\pm k')(\tilde{m} - \tilde{n})( m+n ) \csc^2 \psi - (\pm k') (\tilde{m} - \tilde{n}) k \sec^2 \psi - \\
      &  (m' -n')(\pm \tilde{k})(m+n) \csc^2 \psi + (m' - n')(\pm \tilde{k})k \sec^2 \psi]  a\mathbf{b}^\pm_{{k}'{n}'{n}'}a\mathbf{b}^\pm_{\tilde{k}\tilde{n}\tilde{n}}ab

   \end{array}
\end{equation*}

When we consider the diagonal region where $m'=n'$, $\tilde{m} = \tilde{n}$ and $m=n$, in this case, again we have 
$$
\varepsilon^{ijk} A_iA_jA_k|_{m'=n', \tilde{m}= \tilde{n}, m= n} =0 
$$
Then the residue trace over the diagonal region is 
\begin{equation*}
  \begin{array}{rl}
    &  Res_{z=0}Tr(\varepsilon^{ijk}3A_i[\hat{\partial}_j, A_k] |\mathcal{D}_3|^{-3-z}) \\
    & =  3i    Res_{z=0}Tr(\sum k'k \sec^2 \psi [ \mathbf{a}^+_{p'q'\ell'} \mathbf{b}^-_{k'n'n'}\mathbf{a}^+_{qqk}\mathbf{b}^-_{knn} -  \mathbf{a}^+_{p'q'\ell'} \mathbf{b}^-_{k'n'n'}\mathbf{a}^-_{qqk}\mathbf{b}^+_{knn}\\
    &  +  \mathbf{a}^-_{p'q'\ell'} \mathbf{b}^+_{k'n'n'}\mathbf{a}^-_{qqk}\mathbf{b}^+_{knn}  -  \mathbf{a}^-_{p'q'\ell'} \mathbf{b}^+_{k'n'n'}\mathbf{a}^+_{qqk}\mathbf{b}^-_{knn} ]  |\mathcal{D}_3|^{-3-z}) \\
     & = 3 i   \sum k'k [ {a}_{q'q'k'} {b}'_{k'n'n'}{a}_{qqk}{b}'_{knn} -  {a}_{q'q'k'} {b}'_{k'n'n'}{a}'_{qqk}{b}_{knn}\\
    &  +  {a}'_{q'q'k'} {b}_{k'n'n'}{a}'_{qqk}{b}_{knn}  -  {a}'_{q'q'k'} {b}_{k'n'n'}{a}_{qqk}{b}'_{knn} ] \\
     & Res_{z=0}\sum_m (m+1)^2   (m^2+2m )^{-\frac{3+z}{2}} \\
         & =  6i \sum k'k  [ {a}_{q'q'k'} {b}'_{k'n'n'}{a}_{qqk}{b}'_{knn} -  {a}_{q'q'k'} {b}'_{k'n'n'}{a}'_{qqk}{b}_{knn}\\
    &  +  {a}'_{q'q'k'} {b}_{k'n'n'}{a}'_{qqk}{b}_{knn}  -  {a}'_{q'q'k'} {b}_{k'n'n'}{a}_{qqk}{b}'_{knn} ] \\
   & Res_{z=0} \sum_{k \geq 0} \frac{\Gamma(k+ (3+z)/2)}{\Gamma((3+z)/2) k!}   \zeta_R(z+2k+1)  \\
         & =   6i  \sum k'k [ {a}_{q'q'k'} {b}'_{k'n'n'}{a}_{qqk}{b}'_{knn} -  {a}_{q'q'k'} {b}'_{k'n'n'}{a}'_{qqk}{b}_{knn}\\
    &  +  {a}'_{q'q'k'} {b}_{k'n'n'}{a}'_{qqk}{b}_{knn}  -  {a}'_{q'q'k'} {b}_{k'n'n'}{a}_{qqk}{b}'_{knn} ] \\
  \end{array}
\end{equation*}
 In the above computation, we use $\sec^2 \psi = 1 + \tan^2 \psi$ and $\tan^2 \psi$ is canceled out by the orthogonal basis as before.
\end{proof}

We have seen that these three spectral triples are  all related to the round metric, comparison between their Chern--Simons actions 
 confirms  that
the  Chern--Simons action is not a topological invariant, that is,  it depends on the choice of Dirac operators.  
\begin{prop}
   The noncommutative Chern--Simons action on the quantum 3-sphere $S^3_\theta$ depends on the choice of Dirac operators.
\end{prop}

In the tangent space $T_eS^3 \cong \mathfrak{su}(2) $, we can choose different combinations of vector fields as its basis.
As a result, we have different Dirac operators according to such choice of bases. By the definition of connection 1-forms in noncommutative geometry,
its Chern--Simons action generally depends on the choice of Dirac operators. So how to extract an interesting 
invariant from the noncommutative Chern--Simons theory is still an open problem.
In spite of its complicated form, the Chern--Simons action of the generalized 
Dirac geometry in \eqref{NCCSaction} is more fundamental compared to the other two since it has the dependence on the parameter $\lambda$. Formally, one could continue to
do the path integral quantization, and compute the partition function of the corresponding quantum Chern--Simons theory as in \cite{P1202}.

\section{Acknowledgements}
This project started with suggestions from Walter van Suijlekom while the author was visiting IHES, 
and main parts of this paper  were written down at MPIM.
The author would like to thank IHES and MPIM for their hospitalities, and
 thank  Walter van Suijlekom, Marcolli Marcolli and Alain Connes for their helpful discussions. 


\bibliographystyle{plain}
\bibliography{NCCSS3}

\begin{thebibliography}{10}

\bibitem{CF94}
A.~Chamseddine and J.~Fr{\"o}hlich.
\newblock The {C}hern-{S}imons action in noncommutative geometry.
\newblock {\em J. Math. Phys.}, 35:5195--5218, 1994.

\bibitem{CS74}
S.~S. Chern and J.~Simons.
\newblock Characteristic forms and geometric invariants.
\newblock {\em Ann. of Math.}, (2)99:48--69, 1974.

\bibitem{C94}
A.~Connes.
\newblock {\em Noncommutative Geometry}.
\newblock Academic Press, 1994.

\bibitem{C04}
A.~Connes.
\newblock Cyclic cohomology, quantum group symmetries and the local index
  formula for ${SU}_q(2)$.
\newblock {\em J. Inst. Math. Jussieu}, 3:17–68, 2004.

\bibitem{CC06}
A.~Connes and A.~Chamseddine.
\newblock Inner fluctuations of the spectral action.
\newblock {\em J. Geom. Phys.}, 57:1--21, 2006.

\bibitem{CD02}
A.~Connes and M.~Dubois-Violette.
\newblock Noncommutative finite-dimensional manifolds {I}: {S}pherical
  maniflods and related examples.
\newblock {\em CMP}, 230:539--579, 2002.

\bibitem{CL01}
A.~Connes and G.~Landi.
\newblock Noncommutative maniflods, the instanton algebra and isospectral
  deformations.
\newblock {\em CMP}, 221:141--159, 2001.

\bibitem{CM95}
A.~Connes and H.~Moscovici.
\newblock The local index formula in noncommutative geometry.
\newblock {\em GAFA}, 5:174--243, 1995.

\bibitem{D03}
L.~Dabrowski.
\newblock The garden of quantum spheres.
\newblock In P.~Hajac and W.~Pusz, editors, {\em Noncommutative geometry and
  quantum groups}, pages 37--48. PWN, 2003.

\bibitem{H06}
N.~Higson.
\newblock The residue index theorem of {C}onnes and {M}oscovici.
\newblock In J.~Roe and N.~Higson, editors, {\em Surveys in noncommutative
  geometry}, pages 71--126. AMS, 2006.

\bibitem{H74}
N.~Hitchin.
\newblock Harmonic spinors.
\newblock {\em Advances in Mathematics}, 14:1–55, 1974.

\bibitem{H00}
Y.~Homma.
\newblock A representation of {S}pin(4) on the eigenspinors of the {D}irac
  operator on ${S}^3$.
\newblock {\em Tokyo J. Math.}, 23:453--472, 2000.

\bibitem{K98}
T.~Krajewski.
\newblock Gauge invariance of the {C}hern-{S}imons action in noncommutative
  geometry.
\newblock In D.~Kastler, M.~Rosso, and T.~Schucker, editors, {\em Quantum
  groups, noncommutative geometry and fundamental physical interactions}, pages
  21--35. Nova Science Pub Inc, 1999.

\bibitem{M91}
K.~Matsumoto.
\newblock Noncommutative three dimensional spheres.
\newblock {\em Japan J. Math.}, 17:333--356, 1991.

\bibitem{M11}
F.~Meier.
\newblock Eigenspaces of the spin {D}irac operator over ${S}^3$.
\newblock 2011.
\newblock arXiv: 1103.4097.

\bibitem{NO97}
T.~Natsume and C.~Olsen.
\newblock Toeplitz operators on noncommutative spheres and an index theorem.
\newblock {\em Indiana Univ. Math. J.}, 46:1055--1112, 1997.

\bibitem{P1202}
O.~Pfante.
\newblock {C}hern-{S}imons theory for the noncommutative 3-{T}orus
  ${C}^\infty(\mathbb{T}_\theta^3)$.
\newblock {\em J. Geom. Phys.}, 63:32--44, 2013.

\bibitem{P1201}
O.~Pfante.
\newblock A {C}hern-{S}imons action for noncommutative spaces in general with
  the example ${SU}_q(2)$.
\newblock {\em J. Noncommut. Geom.}, 8:611--654, 2014.

\bibitem{R93}
M.A. Rieffel.
\newblock Deformation quantization for actions of $\mathbb{R}^d$.
\newblock {\em Memoirs AMS}, 506, 1993.
\newblock Providence, RI.

\bibitem{SDLSV05}
W.~van Suijlekom, L.~Dabrowski, G.~Landi, A.~Sitarz, and J.~C. V{\'a}rilly.
\newblock The local index formula for ${SU}_q (2)$.
\newblock {\em K-theory}, 35, 3-4:375--394, 2005.

\bibitem{W89}
E.~Witten.
\newblock Quantum field theory and the {J}ones polynomial.
\newblock {\em CMP}, 121:351--399, 1989.

\end{thebibliography}

\end{document}